\documentclass[aps,pra,superscriptaddress,reprint]{revtex4-1}

\usepackage{amsmath}
\usepackage{amssymb}
\usepackage{fullpage}
\usepackage{amsthm}
\usepackage{graphicx}
\usepackage[capitalize]{cleveref}
\usepackage{xcolor}
\usepackage{tikz}

\newcommand{\id}{\mathbb{I}}

\newcommand{\ket}[1]{|#1\rangle}

\newcommand{\mc}[1]{\mathcal{#1}}

\newcommand{\ct}{^\dagger}
\newcommand{\tn}[1]{^{\otimes #1}}
\newtheorem{theorem}{Theorem}

\newcounter{notecounter}

\usepackage{verbatim}
\usepackage{qcircuit}
\usepackage{float}
\usepackage{appendix}
\usepackage{graphicx}
\usepackage{natbib}
\usepackage[export]{adjustbox}
\usepackage{epstopdf}
\usepackage{tikz-cd}
\usepackage{array}
\usepackage[caption=false]{subfig}
\usepackage{enumerate}
\usepackage{verbatim}
\usepackage{multirow}
\usepackage{algorithm}
\usepackage{algpseudocode}

\newcommand{\code}[1]{\texttt{#1}}
\newcommand{\boardstate}{\code{BOARDSTATE}~}

\newcommand{\sw}{$\sqrt{\mathrm{SWAP}}$~}
\newcommand{\cphase}{$\mathrm{CPHASE}$~}
\newcommand{\cnot}{$\mathrm{CNOT}$~}
\newcommand{\ns}{\mathrm{ns}}
\newcommand{\controlb}{*!<0em,.025em>-=-<.08em>{\blacklozenge}}
\newcommand{\ctrlb}[1]{\controlb \qwx[#1] \qw}

\renewcommand{\mod}[1]{~\mathrm{mod}~#1}
\newcommand{\brif}{\;\;\;\;\;\;\text{if}\;\;\;\;\;}

\crefname{algocf}{alg.}{algs.}
\Crefname{algocf}{Algorithm}{Algorithms}
\crefname{algocfline}{alg.}{algs.}
\Crefname{algocfline}{Algorithm}{Algorithms}

\begin{document}
\title{Quantum error correction in crossbar architectures}
\author{Jonas Helsen}
\affiliation{QuTech, Delft University of Technology, Lorentzweg 1, 2628 CJ Delft, The Netherlands}
\author{Mark Steudtner}
\affiliation{QuTech, Delft University of Technology, Lorentzweg 1, 2628 CJ Delft, The Netherlands}
\affiliation{Instituut-Lorentz, Universiteit Leiden, P.O. Box 9506, 2300 RA Leiden, The Netherlands}
\author{Menno Veldhorst}
\affiliation{QuTech, Delft University of Technology, Lorentzweg 1, 2628 CJ Delft, The Netherlands}
\affiliation{Kavli Institute of Nanoscience, Delft University of Technology, P.O. Box 5046, 2600 GA Delft, The Netherlands.}
\author{Stephanie Wehner}
\affiliation{QuTech, Delft University of Technology, Lorentzweg 1, 2628 CJ Delft, The Netherlands}
\date{\today}
\begin{abstract}
\noindent A central challenge for the scaling of quantum computing systems is the need to control all qubits in the system without a large overhead. A solution for this problem in classical computing comes in the form of so called crossbar architectures. Recently we made a proposal for a large scale quantum processor~[Li et al. arXiv:1711.03807 (2017)] to be implemented in silicon quantum dots. This system features a crossbar control architecture which limits parallel single qubit control, but allows the scheme to overcome control scaling issues that form a major hurdle to large scale quantum computing systems. In this work, we develop a language that makes it possible to easily map quantum circuits to crossbar systems, taking into account their architecture and control limitations. Using this language we show how to map well known quantum error correction codes such as the planar surface and color codes in this limited control setting with only a small overhead in time. We analyze the logical error behavior of this surface code mapping for estimated experimental parameters of the crossbar system and conclude that logical error suppression to a level useful for real quantum computation is feasible.
\end{abstract}
\maketitle

\section{Introduction}\label{sec:introduction}

When attempting to build a large scale quantum computing system a central problem, both from experimental and theoretical perspectives, is what might be called the interconnect problem. This problem, which also exists in classical computing, arises when computational units (e.g.~qubits in quantum computers, transistors in classical computers) are densely packed such that there is not enough room to accommodate individual control lines to every unit. A solution to this problem, which is commonplace in classical computing systems, is a so called `crossbar architecture'. In this class of computing architecture we do not draw a control line to every qubit but rather organize computational units in a grid with control lines addressing full rows and columns of this grid. Control effects then happen at the intersection of column and row lines. In this way, using $N$ control lines $O(N^2)$ computational units can be addressed. This makes it possible to scale the system to a large number of qubits. The price to pay for this is a reduced ability to perform operations on different units in the grid in parallel. For classical systems this is not a fundamental problem, but when the computational units are qubits, whose information decays over time, parallelism becomes absolutely essential. This introduces a formidable roadblock for the development of crossbar systems for quantum computing systems. Nevertheless various crossbar architectures for quantum computers have been proposed in the past~\cite{colless2014control,vandersypen2017interfacing,hill2015surface,veldhorst2017crossbar,veldhorst2016silicon}. Recently~\cite{veldhorst2017crossbar}~we proposed a quantum computing platform based on spin qubits in silicon quantum dots featuring a crossbar architecture. This architecture features compatibility with modern silicon manufacturing techniques and in combination with recent advances in controlling quantum dot qubits and the inherent long coherence times of spin qubits in silicon we expect it to be a formidable step forwards in creating large scale quantum computing devices.\\

\noindent Any realistic quantum computing device, including the one we propose in~\cite{veldhorst2017crossbar}, will suffer from noise processes that degrade quantum information. This noise can be combated by quantum error correction~\cite{gottesman1998theory,lidar2013quantum}, where quantum information is encoded redundantly in such a way that errors can be diagnosed and remedied as they happen without disturbing the encoded information. Many quantum error correction codes have been developed over the last two decades and several of them have desirable properties such as high noise tolerance, efficient decoders and reasonable implementation overhead. Of particular note are the planar surface~\cite{dennis2002topological} and color codes~\cite{bombin2006topological}, which have the nice property that they can be implemented in quantum computing systems where only nearest-neighbor two-qubit gates are available.\\

\noindent However these codes, and all other quantum error correction codes, were developed under the (often implicit) assumption that all physical qubits participating in the code can be controlled individually and in parallel. For large (read: comprising many qubits) error correction codes this introduces a tension between the needs of the error correction code and the control limitations for large systems mentioned above. While practical large-scale quantum computers most likely pose control limitations, surprisingly little work has been done in this area~\cite{versluis2017scalable}. Here we investigate the minimal amount of parallel control resources needed for quantum error correction and focus in particular on crossbar architectures. In \cref{fig:layout,fig:spurious} we summarize the layout and control limitations of the architecture in \cite{veldhorst2017crossbar}. Overcoming these limitations motivates the current work.

\subsection{Contributions}\label{subsec:contributions}

\noindent{\bf Analysis of the crossbar system}\\

\noindent We analyze the crossbar architecture we propose in~\cite{veldhorst2017crossbar}. We give a full description of the layout and control characteristics of the architecture in a manner accessible to non-experts in quantum dots. We develop a language for describing operations in the crossbar system. Of particular interest here are the regular patterns (see e.g.~\cref{subsec:configurations}) that are implied by the crossbar structure. These configurations provide an abstraction on which we build mappings of quantum error correction codes (see below) This analysis is particular to the system in~\cite{veldhorst2017crossbar} but we believe many of the considerations to hold for more general crossbar architectures.\\

\noindent {\bf An efficient algorithms for control on crossbar architectures}\\

\noindent We develop an algorithm for moving around qubits (shuttling) on crossbar architectures. We show that the task of shuttling qubits in parallel can be described using a matrix taking value in an idempotent monoid. The control algorithm then reduces to finding independent columns of this matrix, for a suitable notion of independence. This algorithm in principle allows the straightforward mapping of more complicated quantum algorithms which require long-range operations, with little operational overhead. We also expect this algorithm to be applicable to the control of more general crossbar architectures. We also sketch an algorithm for parallel two-qubit interactions in crossbar systems which produce \emph{optimal} control sequences. This algorithm is based on computing the Schmidt-normal form of matrices with entries in the principal ideal domains $\mathbb{Z}_2$ and $\mathbb{Z}_4$.\\

\noindent {\bf Mapping of surface and color codes}\\

\noindent We map the planar surface code and the $6.6.6.$ (hexagonal) and $4.8.8.$ (square-octagonal) color codes~\cite{bombin2006topological} to the crossbar architecture, taking into account its limited ability to perform parallel quantum operations. The tools we develop for describing the mapping, in particular the configurations described in \cref{subsec:configurations}, should be generalizable to other quantum error correction codes and general crossbar architectures.\\

\noindent {\bf Analysis of the surface code logical error}\\

\noindent Due to experimental limitations the mappings mentioned above might not be attainable in near term devices. Therefore we adapt the above mappings to take into account practical limitations in the architecture~\cite{veldhorst2017crossbar}. In this version of the mapping the length of an error correction cycle scale with the distance of the mapped code. This means the mapping does not allow for arbitrary logical error rate suppression. Therefore we analyze the behavior of the logical error rate with respect to estimated experimental error parameters and find that the logical error rate can in principle be suppressed to below $10^{-20}$ (an error rate comparable to the error rate of classical computers~\cite{tezzaron2004soft}), allowing for practical quantum computation to take place.\\

\noindent Our work raises several interesting theoretical questions regarding the mapping of quantum algorithms to limited control settings, see \cref{sec:conclusion}.

\begin{figure*}[t]
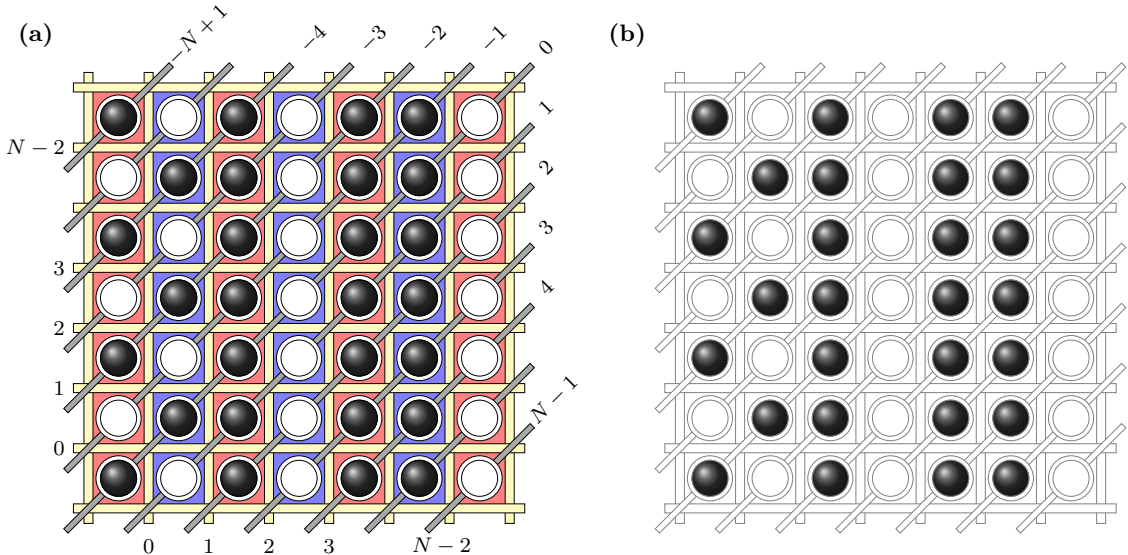

\include{layout}
\caption{ \textbf{(a)} A schematic of the Quantum Dot Processor (QDP) that we propose~\cite{veldhorst2017crossbar}, see \cref{subsec:layout} for details. The white circles correspond to quantum dots, with the black filling denoting the presence of electrons, whose spins are employed as qubits. All dots are embedded in either a red or a blue column. Single qubit gates can only be applied globally on either all qubits in all blue columns or all qubits in all red columns.  The vertical, horizontal (both yellow) and diagonal lines (gray) are a feature of this crossbar scheme.  The horizontal and vertical gate lines implement barriers that isolate the dots from each other. The diagonal lines simultaneously control the dot potentials of all dots coupled to one line. Quantum operations are effected by pulsing individuals lines. In order to perform two qubit operations on qubits in adjacent dots, one typically  needs to lower the barrier that separates them, and change the dot potentials by operating the diagonal lines. Note that two-qubit gates applied to adjacent qubits in the same column are inherently different (by nature of the QDP design) from two-qubit gates between two adjacent qubits in the same row. With the control lines, we can also move qubits from dot to dot and measure them. However, since each control line influences $O(N)$ qubits, individual qubit control, as well as parallel operation on many qubits is limited.  \textbf{(b)} Abstracted version of the QDP scheme representing the classical \boardstate matrix.  The \boardstate holds no quantum information, but encodes where qubits are located on the QDP grid.}\label{fig:layout}
\end{figure*}

\subsection{Outline}
In~\cref{sec:the quantum dot processor} we introduce the architecture we proposed in~\cite{veldhorst2017crossbar}. We forgo an explanation of the physics and focus on the abstract control aspects of the system (explaining them in a largely self-contained manner accessible to non experts in quantum dot physics). We introduce classical helper objects such as the \boardstate which will aid later developments. We discuss one- and two-qubit operations, measurements, and qubit shuttling. In section~\cref{sec:An assembly language for crossbar control} we focus on parallel operations. We discuss difficulties inherent in parallel operation in a crossbar system and develop an algorithm for dealing with them efficiently. We also introduce several \boardstate configurations which feature prominently in quantum error correction mappings and describe how to reach them efficiently by parallel shuttling. In~\cref{sec:error correction codes} we give a quick introduction to quantum error correction with a particular focus on the planar surface code and the $4.8.8.$ and $6.6.6.$ color code. In~\cref{sec:surface code implementation} we bring together all previous sections and devise a mapping of the planar surface code to the crossbar architecture. This we continue in~\cref{sec:color code implementation} for the $6.6.6.$ and $4.8.8.$ color codes. Finally in~\cref{sec:discussion} we analyze in detail the logical error probability of the surface code mapping as a function of the code distance and estimated error parameters of the crossbar system. 

\section{The quantum dot processor}\label{sec:the quantum dot processor}

In this section we will give an overview of the quantum dot processor (QDP) architecture as proposed in~\cite{veldhorst2017crossbar}. We will use this architecture as a concrete realization of the more general idea of quantum crossbar architectures. We will focus not so much on the details of the implementation but rather focus on abstract operational properties of the system as they are relevant for our purposes. The basic organization of the QDP is that for an $N\times N$ grid of qubits interspersed with control lines that effect operations on the qubits. The most notable feature of the QDP (and crossbar architectures in general) is the fact that any classical control signal sent to a control line will be applied simultaneously to all qubits adjacent to that control line. This means that every possible classical instruction applied to the QDP will affect $O(N)$ qubits (these qubits will not necessarily be physically close to each other). This has important consequences for the running of quantum algorithms on the QDP (or any crossbar architecture) that must be taken into account when compiling these algorithms to hardware level instructions. Notably it places strong restrictions on performing quantum operations in parallel on the QDP. To deal with these restrictions it is important to have a good understanding of how operations are performed on the QDP. It is for this reason that we begin our study of the QDP with an examination of its control structure at the hardware level. We describe the physical layout of the system and develop nomenclature for the fundamental control operations. This nomenclature might be called the `machine code' of the QDP. From these basic instructions we go on to construct all elementary operations that can be applied to qubits in the QDP. These are quantum operations, such as single qubit gates, nearest-neighbor two-qubit gates and qubit measurements but also a non-quantum operation called coherent shuttling which does not affect the quantum state of the QDP qubits but changes their connectivity graph (i.e.~which qubits can be entangled by two-qubit gates). All of these operations are restricted by the nature of the control architecture in a way that gives rise to interesting patterns (\cref{subsec:configurations}) and which we will more fully examine in \cref{sec:An assembly language for crossbar control}.

\subsection{Layout}\label{subsec:layout}

A schematic overview of the QDP architecture is given in~\cref{fig:layout}, where qubits (which are electrons, denoted by black balls) occupy an array of $N\times N$ quantum dots (hereafter often referred to as sites). The latter are denoted by white sites when empty, since they either are occupied by a qubit or not. We will label the dots by tuples containing row and column indices $(i, j) \in [0:N-1]^{\times 2}$ (beginning from the \emph{bottom left} corner), such that a single qubit state $\ket{\psi}$ living on the $(i,j)$'th site will be denoted by $\ket{\psi}_{(i,j)}$. We assume the qubits to be initialized in the state $\ket{0}$. For future reference we note that $\ket{0}$ corresponds to the spin-up state and $\ket{1}$ to the spin-down state of the electron constituting the qubit.\\

\noindent Typically we will work in a situation where half the sites are occupied by a qubit and half the sites are empty (as seen in~\cref{fig:layout} (a)). Because (as we discuss in \cref{subsubsec:shuttling}) the qubits can be moved around on the grid and the two-qubit gates depend on the filling of the grid, it is important to keep track of which sites contain qubits and which ones do not. This can be done efficiently in classical side-processing. To this end we introduce the \boardstate object. \boardstate consists of a binary $N\times N $ matrix with a $1$ in the $(i,j)$'th place if the $(i,j)$'th site contains an electron (qubit) and a $0$ otherwise. The \boardstate does not contain information about the qubit state $\ket{\psi}_{(i,j)}$, only about the electron occupation of the grid. A particular \boardstate is illustrated in the left panel of~\cref{fig:layout}.\\

\noindent We now turn to describing the control structures that are characteristic for this architecture. As a first feature, we would like to point out that each site is either located in a red or a blue region in~\cref{fig:layout} (left panel). The blue (red) columns correspond to regions of high (low) magnetic fields, which plays a role in the addressing of qubits for single qubit gates. We will denote the set of qubits in blue columns (identified by their row and column indices) by $\mc{B}$ and the set of qubits in red columns by $\mc{R}$.\\

\noindent Much finer groups of sites can be addressed by the control lines that run through the grid. The crossbar architecture features control lines that are connected to $O(N)$ sites. At the intersections of these control lines individual sites and qubits can be addressed. This means that using $O(N)$ control lines $O(N^2)$ qubits can be controlled. As seen in~\cref{fig:layout} the rows and columns of the QDP are interspersed with horizontal and vertical lines (yellow), as a means to control the tunnel coupling between adjacent sites. We refer to those lines as barrier gates, or barriers for short. Each line can be controlled individually, but a pulse has an effect on all $O(N)$ qubits adjacent to the line. Another layer of control lines is used to address the dots itself rather than the spaces in between them. The diagonal gate lines (gray), are used to regulate the dot potential. We label the horizontal and vertical lines by an integer running from $0$ to $N-2$ and the diagonal lines with integers running from $-(N-2)$ to $N-2$ where the $-(N-2)$'th line is the top-left line and increments move towards the bottom right (see~\cref{fig:layout}(a)). Next we describe how these control lines can be used to effect operations on the qubits occupying the QDP grid.

\subsection{Control and addressing}\label{subsec:control and operations}
As described above, the QDP consists of quantum dots interspersed with barriers and connected by diagonal lines. For our purposes these can be thought of as abstract control knobs that apply certain operations to the qubits. In this section we will describe what type of gates operations are possible on the QDP. We will not concern ourselves with the details of parallel operation until \cref{sec:An assembly language for crossbar control}.\\

\noindent There are three fundamental operations on the QDP which we will call the ``grid operations". These operations are ``lower vertical barrier" (\code{V}), ``lower horizontal barrier" (\code{H}) and ``set diagonal line" (\code{D}). The first two operations are essentially binary (on-off) but the last one (\code{D}) can be set to a value $t\in [0:T]$ where $T$ is a device parameter. (At the physical level this corresponds to how many clearly distinct voltages we can set the quantum dot plunger gates~\cite{veldhorst2017crossbar}). Although the actual pulses on those gates differ by amplitude and duration between the different gates and operations, this notation gives us a clear idea which lines are utilized. This can be done because realistically one will not interleave processes in which pulses have such different shapes. 
We can label the grid operations by mnemonics (which in a classical analogy we will call OPCODES) as seen in~\cref{fig:grid operations}. These OPCODES are indexed by an integer parameter that indicates which control line it applies to. We count horizontal and vertical lines starting at zero from the lower left corner of the grid (see \cref{fig:layout}). Note that the lines at the boundary of the grid are never adressed in our model and are thus not counted.\\

We indicate parallel operation of a collection of OPCODES by ampersands, e.g.~\code{D[1]\&H[2]\&D[5]}.
We also define inherently parallel versions (in \cref{fig:parallel grid operations}) of the basic OPCODES that take as input a binary vector \code{V} of length $N$ (for the diagonal line this is a $T$-valued vector of length $N$)
\begin{figure}[ht]\label{fig:grid operations}
\centering
\begin{tabular}{|c|c|}
\hline
OPCODE & Effect \\
\hline
\code{V[i]} & Lower vertical barrier at index \code{i} \\
\code{H[i]} & Lower horizontal barrier at index \code{i}\\
\code{D[i][t]} & Set diagonal line at index \code{i} to value \code{t}\\
\hline
\end{tabular}
\end{figure}
\begin{figure}[ht]\label{fig:parallel grid operations}
\centering
\begin{tabular}{|c|c|}
\hline
OPCODE &Effect \\
\hline
\code{V[V]} & Set vertical barrier to \code{V}$(i)$, $\forall i\!\in \! [0 \!:\!  N \!-  \!2 ]$\\
\code{H[V]} & Set horizontal barrier to \code{V}$(i)$, $\forall i\!\in\! [0\!:\!\!N\!-\!2]$ \\
\code{D[V]} & Set diagonal at height \code{V}$(i)$, $\forall i\!\in \![-\!N \!\!+\! \!2\!:\!\!N\!\!-\!\!2]$ \\
\hline
\end{tabular}
\end{figure}
These grid operations can be used to induce some elementary quantum gates and operations on the qubits in the QDP. Below we describe these operations.

\subsection{Elementary operations}\label{subsec:elementary operations}
Here we give a short overview of the elementary operations available in the QDP. We will describe basic single qubit gates, two-qubit gates, the ability to move qubits around by coherent shuttling~\cite{fujita2017coherent} and a measurement process through Pauli Spin Blockade (PSB)~\cite{hanson2007spins}. All of these operations are implemented by a combination of the grid operations defined in \cref{fig:grid operations}, and always have a dependence on the \boardstate.

\subsubsection{Coherent qubit shuttling}\label{subsubsec:shuttling}
An elementary operation of the QDP is the coherent qubit shuttling~\cite{fujita2017coherent,taylor2005fault}, of one qubit to an adjacent, empty site. That means that an electron (qubit) is physically moved to the other dot (site) utilizing at least one diagonal line and the barrier between the two sites. It thereby does not play a role whether the shuttling is in horizontal (from a red to a blue column or the other way around) or vertical direction (inside the same column). However, the shuttling in between columns results in a $Z$ rotation, that must be compensated by timing operations correctly, see~\cite{veldhorst2017crossbar} for details. This $Z$ rotation can also by used as a local single qubit gate, see~\cref{subsubsec:single qubit rotations}. The operation is dependent on the \boardstate by the prerequisite that the site adjacent to the qubit to must be empty. Collisions of qubits are to be avoided, as those will lead to a collapse of the quantum state (see however the measurement process in~\cref{subsubsec:measurement}). We now describe the coherent shuttling as the combination of grid operations. \\

\noindent We lower the vertical (or horizontal) barrier in between the two sites and instigate a `gradient' of the on-site potentials of the two dots. That is, the diagonal line of the site containing the qubit must be operated at $t\in [0:T]$ while the line overhead the empty site must have the potential $\hat{t}\in [0:T]$ with $\hat{t}=t-1$. Note that this implies it might not be operated at all (if it is already at the right level). We will subsequently refer to the combination of a lowered barrier and such a gradient as a ``flow". A flow will in general be into one of the four directions on the grid. We define the commands \code{VS[i,j,k]} (vertical shuttling) and \code{HS[i,j,k]} (horizontal shuttling). The command \code{VS[i,j,k]} shuttles a qubit at location $(i,j)$ to $(i+1,j)$ for $k=1$ (upward flow) and shuttles a qubit at location $(i+1,j)$ to $(i,j)$ for $k=-1$ (downward flow). Similarly, the command \code{HS[i,j,k]} shuttles a qubit at location $(i,j)$ to $(i,j+1)$ for $k=1$ (rightward flow) and shuttles a qubit at location $(i,j+1)$ to $(i,j)$ for $k=-1$ (leftward flow). See \cref{tab:operation opcodes} for a summary of these OPCODES.\\

\noindent Using only these control lines, we can individually select a single qubit to be shuttled. However, when attempting to shuttle in a parallel manner, we have to be carefully take into account the effect that the activation of several of those lines has on other locations. We will deal with this in more detail in \cref{subsec:parallel shuttle operations}.

\begin{figure*}[t]
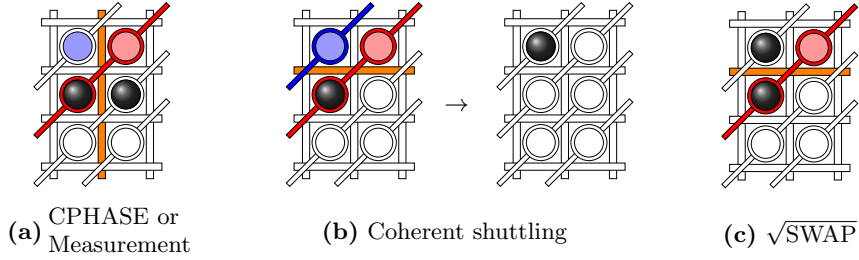

\include{elementary}
\caption{Schematic representation of the use of control lines for the native operations in the QDP.  Qubits are represented by black balls on the grid. Red or blue colored dots are empty, but their dot potentials change due to an operation of the diagonal line they are coupled to.  Empty dots unaffected by  grid operations are white.  \textbf{(a)}  Grid operations necessary to perform a measurement or a two-qubit effective \cphase gate between the two qubits. The orange barrier between the two qubits is lowered, and the dot potentials along the red diagonal line is raised by pulsing the latter. Note that the empty, red colored dot is also effected by that action, and its barrier to the adjacent dot is lowered. If the two dots in the upper row were not empty, side effects would occur. See \cref{subsubsec:multi-qubit rotations} for more information on the nature of the two-qubit gates.  Note also that the readout procedure of the measurement requires us to have the upper dot (light blue) empty, if the barrier gate between them is used for readout. \textbf{(b)} Vertical shuttling of a qubit (to the top dot) requires to lower the orange barrier.  One can than either raise the dot potentials on the red diagonal line, or lower the potential on the blue dot  by addressing  the blue diagonal.   \textbf{(c)} Schematic representation of the control lines used for performing two-qubit \sw gate between the two qubits on that grid. The orange barrier is lowered and the red diagonal line is utilized to detune dot potentials.  }
\label{fig:basic operations}
\end{figure*}

\subsubsection{Measurement and readout}\label{subsubsec:measurement}
The QDP allows for local single qubit measurements in the computational basis $\ket{0},\ket{1}$. We can measure a qubit by attempting to shuttle it to a horizontally adjacent site that is already occupied by an ancilla qubit and then detecting whether the shuttling was successful. This process is called Pauli Spin Blockade (PSB) measurement~\cite{hanson2007spins,veldhorst2017crossbar}. However, the QDP's ability to perform this type of qubit measurements is limited by three factors.\\

\noindent Firstly, the measurement requires an ancilla qubit horizontally adjacent to the qubit to be measured. This ancilla qubit must be in a known computational basis state. Moreover, if the ancilla qubit is in the state $\ket{0}$ the ancilla qubit must be in the set $\mc{B}$ (blue columns in \cref{fig:layout}) while the qubit to be measured must be in the set $\mc{R}$ (red columns in \cref{fig:layout}). On the other hand, if the ancilla qubit is in the state $\ket{1}$ the ancilla qubit must be in the set $\mc{R}$ while the qubit to be measured is in the set $\mc{B}$. This means that when an qubit-ancilla pair is in the wrong configuration we must first shuttle both qubits one step to the left (or both the the right). Note that this takes two additional shuttling operations, which means it is important to keep track at all times where on the \boardstate the qubit and its ancilla are or else incur a shuttling overhead (which might become significant when dealing with large systems and many simultaneous measurements). We will deal with this problem of qubit-ancilla pair placement in more detail in \cref{subsec:parallel measurements}. \\

\noindent Secondly, assuming that the qubit-ancilla pair is in the right configuration to perform the PSB process one still needs to perform a shuttling-like operation to actually perform the measurement. On the technical level, the operation is different from coherent shuttling, but the use of the lines is similar with the difference that after the readout, the shuttling-like operation is undone by the use of the same lines as before - which are not necessarily the lines one would use to reverse a coherent shuttling operation.  However, scheduling measurement events on the QDP is at least as hard as the scheduling of shuttle operations discussed above. Depending on the state the qubit is in, it will now assume one of two possible states that can be distinguished by their charge distribution.\\

\noindent Thirdly, the readout process requires to have a barrier line that borders to the qubit pair, with an empty dot is across the spot of the qubit to be measured. This is a consequence of the readout procedure.

In \cref{tab:operation opcodes} we introduce the measurement OPCODE \code{M[i,j,k]} with $k \in \{−1, 1\}$ to denote a measurement of a qubit at location $(i,j)$ with an ancilla located to the left $(k=−1)$ or to the right $(k=1)$.

\begin{table*}[t]
\centering
\begin{tabular}{|c|c|c|}
\hline
OPCODE & Control OPCODES & Effect\\
\hline
\multirow{2}{*}{\code{HS[i,j,k]}} & \code{V[i]\&D[i-j][t-1/2-k/2]} & $(k=1)$: ~~Shuttle from $(i,j)$ to $(i,j+1)$\\
& $\;\;\;\;\;\;\;\;\;$\code{\&D[i-j+1][t-1/2+k/2]} &$(k=-1)$: Shuttle from $(i,j+1)$ to $(i,j)$ \\
																														\hline
\multirow{2}{*}{\code{VS[i,j,k]}} & \code{H[j]\&D[i-j][t-1/2-k/2]} & $(k=1)$: ~~Shuttle from $(i,j)$ to $(i+1,j)$\\
& $\;\;\;\;\;\;\;\;\;$\code{\&D[i-j-1][t-1/2+k/2]}&$(k=-1)$: Shuttle from $(i+1,j)$ to $(i,j)$\\ 
\hline
\code{M[i,j,k]} & \code{HS[i,j+1/2+k/2,-k]}& Measurement of qubit at $(i,j)$ using the ancilla at $(i,j+k)$ \\
\hline
\end{tabular}
\caption{OPCODES for horizontal and vertical shuttling and measurement together with the control OPCODES required to implement these operations on the QDP. }\label{tab:operation opcodes}
\end{table*}

\subsubsection{Single-qubit rotations}\label{subsubsec:single qubit rotations}
There are two ways in which single qubit rotations can be performed on the QDP, both with drawbacks and advantages. The first method, which we call the semi-global qubit rotation, relies on electron-spin-resonance~\cite{veldhorst2014addressable}. Its implementation in the QDP allows for any rotation in the single qubit special unitary group $SU(2)$~\cite{nielsen2002quantum} to be performed but we do not have parallel control of individual qubits. The control architecture of the QDP is such that we can merely apply the same single qubit unitary rotation on all qubits in either $\mc{R}$ or $\mc{B}$ (even or odd numbered columns). Concretely we can perform in parallel the single qubit unitaries
\begin{align}\label{eq:single qubit rotations}
U_\mc{R} &= \bigotimes_{(i,j) \in \mc{R}} U_{i,j}\hspace{5mm} U\in SU(2)\\
U_\mc{B} &= \bigotimes_{(i,j) \in \mc{B}} U_{i,j}\hspace{5mm} U\in SU(2),
\end{align}
where $U_{i,j}$ means applying the same unitary $U$ to the state carried by the qubit at location $(i,j)$. In general the only way to apply an arbitrary single qubit unitary on a single qubit in $\mc{B}$ (or $\mc{R}$) is by applying the unitary to all qubits in $\mc{B}$ ($\mc{R})$, moving the desired qubit into an adjacent column, i.e.~from $\mc{B}$ to $\mc{R}$ ($\mc{R}$ to $\mc{B}$) and then applying the inverse of the target unitary to $\mc{R}$ ($\mc{B}$). This restores all qubits except for the target qubit to their original states and leaves the target qubit with the required unitary applied. The target qubit can then be shuttled to its original location. A graphical depiction of the \boardstate associated with this manoeuvre can be found in \cref{fig:single unitary boardstate}. This means applying a single unitary to a single qubit takes a constant amount of grid operations regardless of grid size.\\

\noindent The second method does allow for individual single qubit rotations but is limited to performing single qubit rotations of the form
\begin{equation}\label{eq:rotation by waiting}
U(\phi) = e^{i\phi Z},\hspace{10mm}Z = \begin{pmatrix} 1 & 0\\ 0 &-1\end{pmatrix},\;\; \phi \in [0,2\pi)
\end{equation}
This operation can be performed on a given qubit $\ket{\psi}_{(i,j)}$ by shuttling it from $(i,j)$ to $(i,j\pm 1)$. When the qubit leaves the column it was originally defined ($\mc{B}$ to $\mc{R}$ or vice versa) it will effectively start precessing about its $Z$ axis~\cite{veldhorst2017crossbar}. This effect is always present but it can be mitigated by timing subsequent operations such that a full rotation happens between every operation (effectively performing the identity transformation, see \cref{subsubsec:shuttling}). By changing the timing between subsequent operations any rotation of the form \cref{eq:rotation by waiting} can be effected. This technique will often be used to perform the $Z$ gate (defined above) and the $S=\sqrt{Z}$ phase gate in error correction sequences.

\begin{figure*}[t]
\include{single}
\caption{\boardstate schematic for applying the unitary $U$ to a single qubit (red).  Time flows from left to right in the schematic. This process illustrates both, the possibility to retain single qubit control by using coherent shuttling, and the overhead that comes with it.  \textbf{(a)} we firstly apply the unitary $U$ (blue bars) to all qubits in $\mc{R}$ ($\mc{B}$). We then move the qubit  to the adjacent column. Note that this takes two operations because we do not want any other qubits transitioning with it. In \textbf{(b)},  we apply the inverse unitary $U^\dagger$ to all qubits in $\mc{R}$ ($\mc{B}$). In the last step we move the red qubit back, such that it is in its original position in \textbf{(c)}.}\label{fig:single unitary boardstate}
\end{figure*}

\subsubsection{Two-qubit gates}\label{subsubsec:multi-qubit rotations}
As the last elementary tool, we have the ability to apply entangling two-qubit gates on adjacent qubits. For this case, we need to address the barrier between them and at least one diagonal line. The actual gate that is applied when using those lines now differs for horizontal and vertical two-qubit gates. Inside one column, so for gates between qubits $(i,j)$ and $(i\pm 1, j)$, a square-root of SWAP (\sw) can be realized~\cite{petta2005coherent}, which is defined as 

\begin{align}
\sqrt{\mathrm{SWAP}} =\left(
\begin{matrix}
1 \\ 
&\left( 1+i \right)/2 & \left( 1-i \right) / 2 \\
& \left( 1-i \right) / 2 & \left( 1+i \right)/ 2 \\
& & & 1 
\end{matrix}\right) \, ,
\end{align}
in the computational basis. However between horizontally adjacent qubits, e.g.~between $(i,j)\in \mc{R}$ and $(i,j \pm 1)\in \mc{B}$ the native two-qubit gate is rather an effective \cphase gate
\begin{align} 
 \mathrm{CPHASE} = \left( 
\begin{matrix}
1 \\ & e^{i\phi_1} \\ && e^{i\phi_2} \\ &&& 1
\end{matrix} 
 \right) \, ,
 \end{align}
again in the computational basis and with the two angles $\phi_1 + \phi_2 \mod 2\pi = \pi $ (demonstrated in~\cite{meunier2011efficient,watson2017programmable,veldhorst2015two}). In practice we expect the \sw gate to have significantly higher fidelity than the \cphase gate~\cite{veldhorst2017crossbar} so in any application (e.g.~error correction) the \sw gate is the preferred native two-qubit gate on the QDP. In \cref{tab:interaction opcodes} we define OPCODES for the horizontal interaction (\cphase) and the vertical interaction (\sw)

\subsubsection{CNOT subroutine}\label{subsubsec:CNOT subroutines}
Many quantum algorithms are conceived using the \cnot gate as the main two-qubit gate. However the QDP does not support the \cnot gate natively. It is easy to construct the \cnot gate from the \cphase gate by dressing the \cphase gate with single qubit Hadamard rotations as seen in \cref{fig:cnot construction} (left). It is slightly more complicated to construct a \cnot gate using the \sw but it can be done by performing two \sw gates interspersed single qubit rotations~\cite{schuch2003natural,watson2017programmable,veldhorst2015two} as seen in \cref{fig:cnot construction} (right). If the control qubit is moved from an adjacent column on the QDP (as it is in most cases we will deal with) the $Z$ and $S$ gates can be performed by the $Z$-rotation-by-waiting technique described in the last section.
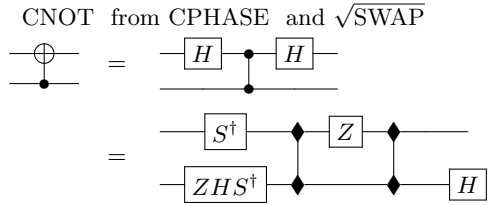
\begin{figure}
\begin{minipage}[t]{0.48\textwidth}
\centering
{{\cnot} from {\cphase} and {\sw}} \\
\begin{tikzpicture}
\draw[white] (0,0)--(10,0);
\node[right] at (1,0) {
\Qcircuit @C=1em @R=.7em{
	 & \targ 	& \qw	\\
	 & \ctrl{-1} & \qw						
}};
\node[right] at (2.3,0) {=};
\node[right] at (3,0) {\Qcircuit @C=1em @R=.7em {
	 &\gate{H} & \ctrl{0}	&\gate{H}& \qw	\\
	&\qw & \ctrl{-1} & \qw& \qw						
}};
\end{tikzpicture}
\end{minipage}
\begin{minipage}[t]{0.48\textwidth}
\centering
$\quad$ \\
\begin{tikzpicture}[baseline=0]
\draw[white] (0,0)--(10,0);
\node[right] at (2.3,0) {=}; 
\node[right] at (3,0) {\Qcircuit @C=1em @R=.7em {
	 &\gate{S\ct} & \ctrlb{0}	&\gate{Z}& \ctrlb{0} & \qw&\qw	\\
	&\gate{ZHS\ct} & \ctrlb{-1} & \qw& \ctrlb{-1} & \qw	&\gate{H}					
}};
\end{tikzpicture}
\end{minipage}
\caption{Construction of the \cnot gate out of the native \cphase and \sw gates. Note that one requires two \sw gates to construct a \cnot gate~\cite{schuch2003natural}. When performing arbitrary algorithms it would be preferable to forgo this substitution and instead compile the algorithm directly into a gateset containing the \sw gate.}\label{fig:cnot construction}
\end{figure}
For completeness we also define an OPCODE for the \cnot operation in \cref{tab:interaction opcodes}.

 \begin{table*}[t]
 \centering
 \begin{tabular}{|c|c|c|}
 \hline
 OPCODE & Effect & Parameter\\
 \hline
 \code{HI[(i,j)]} & Perform \cphase gate between sites \code{(i,j)} and \code{(i,j+1)} & \code{$(i,j) \in [0:N-2]^{\times 2}$}\\
 \code{VI[(i,j)]} & perform \sw gate between sites \code{(i,j)} and \code{(i+1,j)} & \code{$(i,j) \in [0:N-2]^{\times 2}$}\\
 \hline
 \code{HC[(i,j)]} & Perform \cnot (using \cphase) between \code{(i,j)} and \code{(i,j+1)} & \code{$(i,j) \in [0:N-2]^{\times 2}$}\\
 \code{VC[(i,j)]} & perform \cnot (using \sw) between \code{(i,j)} and \code{(i+1,j)} & \code{$(i,j) \in [0:N-2]^{\times 2}$}\\
 \hline
 \end{tabular}
 \caption{OPCODES for horizontal and vertical two-qubit operations on the QDP, respectively the \cphase and \sw gates. We also include OPCODES for the performing of \cnot gates composed of \sw or \cphase gates. }\label{tab:interaction opcodes}
 \end{table*}

\section{Parallel operation of a crossbar architecture}\label{sec:An assembly language for crossbar control}

In this section we focus on performing operations in parallel on the QDP (or more general crossbar architectures). Because of the limitations imposed by the shared control lines of the crossbar architecture, achieving as much parallelism as possible is a non-trivial task. We will discuss parallel shuttle operations, parallel two qubit gates, parallel single qubit gates and parallel measurement. As part of the focus on parallel shuttling we also include some special cases relevant to quantum error correction where full parallelism is possible.\\

\noindent Before we start our investigation however, we would like to put three issues into focus that are likely to be encountered when attempting parallel operations. Firstly, it must be understood that an operation on one location on a crossbar system can cause unwanted side effects in other locations (that might be far away). As indicated in \cref{sec:the quantum dot processor} many elementary operations on the grid in particular take place at the crossing points of control lines. This means that any parallel use of these grid operations must take into account ``spurious crossings'' which may have such unintended side effects. We can illustrate this with an example. Imagine we want to perform the vertical shuttling operations \code{VS[i,j-1,1]} and \code{VS[i+2,j-1,1]} in parallel (see 
\cref{fig:spurious} for illustration). We can do this by lowering the horizontal barriers at rows $i$ and $i+2$ (orange in illustration) and elevating the on-site potentials on the diagonal lines $i-j+1$ and $i+2-j+1$ (red in illustration). This will open upwards flows at locations $(i,j-1)$ and $(i+2,j-1)$. However it will also open an upward flow at the location $(i+2,j+1)$. This means, if a qubit is present at that location an unintended shuttling event will happen. To avoid this outcome we must either perform the operations \code{VS[i,j-1,1]} and \code{VS[i+2,j-1,1]} in sequence (taking two time-steps) or perform an operation \code{VS[i+2,j+1,-1]} to fix the mistake we made, again taking two time-steps. This is a general problem when considering parallel operations on the QDP.

\begin{figure*}[t]
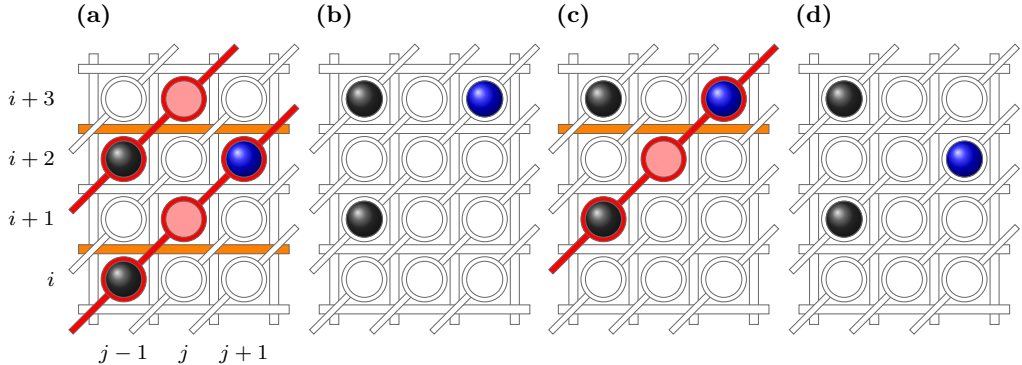

\include{spurious}
\caption{Spurious shuttle operations. Here we illustrate an example of unintended side effects that occur due to the limited control. We again denote qubits by colored balls, and color barriers and lines that are operated. Empty dots with changed potentials are colored as well, white dots are unaffected. \textbf{(a)} The black qubits are to be shuttled from  $(i,j-1)$ to $(i+1,j-1)$ and from $(i+2,j)$  to $(i+3,j)$ respectively without moving the blue qubit. For that purpose, the (orange) barriers between the two dot pairs are lowered, as well as the (red) diagonal lines through $(i,j-1)$ and $(i+2,j)$ are pulsed, such that the dot potentials on those sites are raised. \textbf{(b)} The qubit on $(i+3,j+1)$ has unintentionally moved to $(i+2,j+1)$. \textbf{(c)} To remedy this situation, we lower the barrier number $i+2$ again (orange), and also raise the potential on $(i+3,j+1)$ and all other dots that are connected by the pulsed diagonal line (red).  In \textbf{(d)}, the desired situation is achieved.  } \label{fig:spurious}
\end{figure*}

\noindent Secondly, we would like to point out that in realistic setups, we expect a trade-off between parallelism (manifested in algorithmic depth) and operation fidelity (in particular this will be the case in the QDP system). In order to understand this, we have to be aware that most operations consist of applying the correct pulses for the right amount of time. These durations however can slightly vary from site to site (due to manufacturing imperfections), so we e.g.~must be able to switch barriers back on again prematurely when accounting for a site with a shorter time required. If this is not possible (maybe because it would cause side effects) a loss in operation fidelity is a consequence of the resulting improperly timed operation. The most robust case is thus to schedule operations line-by-line. By this we mean that we attempt to perform $O(N)$ grid operations in a time-step while using every horizontal, diagonal or vertical line only once per individual grid operation. If we for instance schedule several vertical shuttle operations, we may choose to start by lowering one of horizontal barrier first and then detune the dot potentials of all qubits adjacent to that barrier, by pulsing the corresponding diagonal lines. To account for the variations, we reset the diagonal lines at slightly different times. Line-by-line operations work with either line types for every two-dot operation (measurement, shuttling and two-qubit gates). Note however that for shuttling operations individual control over one line is sufficient, whereas for measurement and two-qubit gates we would ideally like to be able to control two lines per qubit pair individually, where one line should be the barrier separating the two paired qubits. Results presented in the following take into account these constraints for quantum error correction. The parallel operation nonetheless remains one of the greatest challenges of the crossbar scheme. In this section we will assume all operations to be perfect (even when performed in parallel) but in \cref{sec:discussion} we perform a more detailed analysis of the behavior of the QDP when operational errors are taken into account.\\ 

\noindent Thirdly, from a performance perspective it is important to separate the operations that have to be done on the qubits on the crossbar grid
from operations that can be done by classical side computation (which for our purposes is essentially free). We will deal with this by including classical side computation in the OPCODES for parallel operation. This way the complexity of dealing with spurious operations is abstracted away. We devise algorithms that take in an arbitrary list of shuttling or two-qubit gate locations and work out a sequence of shuttling or two-qubit gate steps that achieve that list. We begin with discussing parallel shuttle operations.

\subsection{Parallel shuttle operations}\label{subsec:parallel shuttle operations}

\noindent We define parallel versions of the shuttling OPCODES $\code{HS[i,j,k]}$ and $\code{VS[i,j,k]}$ as 

\begin{figure}[ht]\label{fig:parallel shuttling codes}
\centering
\begin{tabular}{|c|c|}
\hline
OPCODE & Effect \\
\hline
\code{HS[L]} & Perform \code{HS[i,j,k]} for all $(i,j,k) \in \code{L}$ \\
\code{VS[L]} & Perform \code{HS[i,j,k]} for all $(i,j,k) \in \code{L}$ \\
\hline
\end{tabular}
\end{figure}
This code takes in a set (denoted as \code{L}) of tuples $(i,j,k)$ which denote `locations at which shuttling happens' $(i,j)$ and `shuttling direction' $(k)$. From these codes it is not immediately clear how many of the shuttling operations can be performed in a single grid operation, i.e.~setting the diagonal lines to some configuration and lowering several horizontal or vertical barrier. If multiple grid operations are needed (such as in the example \cref{fig:spurious}) we would like this sequence of grid operations to be as short as possible. However, given some initial \boardstate and a parallel shuttling command \code{HS[L]} it is not clear what the sequence of parallel shuttling operations actualizing this command is. Below we analyze this problem of parallel shuttling in more detail and give a classical algorithm that produces, from an input \code{HS[L]} or \code{VS[L]} a sequence of parallel grid operations that performs this command. Ideally we would like this sequence to be as short as possible. This algorithm does not perform optimally in all circumstances (i.e.~it does not produce the shortest possible sequence of parallel shuttling operations) but for many relevant cases it performs quite well.
Note that this is a technical section and the details are not needed to understand the quantum error correction results in \cref{sec:error correction codes,sec:surface code implementation,sec:color code implementation}. Readers interested only in those may skip ahead to \cref{subsec:parallel two-qubit gates}

\subsubsection{The flow matrix}
We will only consider shuttling to the left and to the right but all mechanisms introduced work equally well for shuttling in the vertical directions. As will be seen in \cref{subsec:configurations} some \boardstate configurations can be converted into each other in an amount of grid operations that is constant in the size of the grid. It can be seen  that the problem of whether two shuttles can be performed in parallel is a problem with a matrix structure, as flows can only occur at the intersection open barriers and non-trivial diagonal line gradients. To capture this matrix intuition we construct, from the initial \boardstate and the command \code{HS[L]} a matrix $F$ which we call the flow matrix. This matrix will have entries corresponding to the crossing of the gradient line between two diagonal qubit lines and the vertical barrier lines. The flow matrix is defined with respect to a specific command \code{HS[L]} and its entries correspond to the locations on the grid where we want shuttling in certain directions to happen.\\

\begin{figure}[h]
\include{flowmatrix}
\caption{Example of a \boardstate, a parallel command \code{HS[List]} and the corresponding flow matrix $F$.}\label{fig:flow matrix}
\end{figure}

\noindent From a specific command \code{HS[L]} and a specific current \boardstate we will define a flow matrix $F$. This matrix will have entries which take value in the set $\{r,l,e,re,le,*\}$. Each element of this set has a specific operational meaning. The elements $r,l,e$ correspond to specific actions that can be taken on the qubit grid. They correspond specifically to `shuttle to the right' $(r)$, `shuttle to the left' $(l)$ and `do nothing' $(e)$. Note that these actions do not necessarily act on a fixed qubit. Rather they act on a specific location on the grid (where a qubit mayb or may not be present). The other three elements do not directly correspond to a shuttling action but rather signify that at this location we have a choice of different consistent actions. We will call these elements `wildcards'. These wildcards signify the actions `shuttle to the right or do nothing' $(re)$, `shuttle to the left or do nothing' $(le)$, or `any action is allowed' $(*)$. \\

\noindent We fill in the matrix entry $F_{ij}$ with a symbol $r$ for every $(i,j,1)$ in \code{L}. This indicates that at some point in time we want to perform the operation \code{HS[i,j,1]} at that location. Similarly we fill in a symbol $l$ on every matrix entry $F_{ij}$ for every $(i,j+1,-1)$ in \code{L}. We place the symbols $re,le$ respectively on the matrix entries $F_{i(j-1)}$ and $F_{ij}$ for every occupied site $(i,j)$ in the \boardstate that has no corresponding entry in \code{L}. This indicates that we would like for no shuttle operations to happen on these crossing points (since we want the qubit to stay put) but that we do not mind a \code{HS[i,j-1,1]} happening on the crossing point to the left of the qubit at $(i,j)$ (since it will not affect the qubit) or mind a \code{HS[i,j,1]} happening to the right of the qubit at $(i,j)$. Lastly we fill in the symbol $e$ on every matrix entry $F_{ij}$ where we want no shuttling operation to happen at any time to the right of the site $(i,j)$ (for instance on the crossing point between two qubits that are in horizontally adjacent sites). In every other matrix entry $F_{ij}$ we fill in the wildcard symbol $*$ indicating that we do not care if any operation happens at this crossing point. Let's summarize the above construction by 
\begin{widetext}
\begin{equation*}
F_{ij} = \begin{cases} r \;\brif (i,j,1) \in \code{L}\\
						l \;\:\brif (i,j,1) \in \code{L}\\
						e \;\brif (\boardstate(i,j) =1\land \boardstate(i,j+1) =1)\land((i,j,k) \not\in \code{L},k\in \{1,-1\}) \\
						re \brif (\boardstate(i,j) =0\land \boardstate(i,j+1) =1)\land((i,j,k) \not\in \code{L},k\in \{1,-1\})\\
						le \brif (\boardstate(i,j) =1\land \boardstate(i,j+1) =0)\land((i,j,k) \not\in \code{L},k\in \{1,-1\})\\
						* \brif (\boardstate(i,j) =0\land \boardstate(i,j+1) =0)\land((i,j,k) \not\in \code{L},k\in \{1,-1\}).
						\end{cases}
\end{equation*}
\end{widetext}
The flow matrix $F$ takes values in the set $\{r,l,e,re,le,*\}$. In \cref{appsec:Shuttling algorithm} we discuss the mathematical structure of this set in more detail. The above construction gives us a matrix of operations we would like to apply to the initial \boardstate. You can see an example of a \boardstate and \code{HS[L]} command with corresponding flow matrix $F$ in~\cref{fig:flow matrix}.

\subsubsection{An algorithm for parallel shuttling}

 The task is now to subdivide the flow matrix $F$ into a sequence of shuttling operations that can be performed in parallel. Ideally we would like this sequence to be as short as possible. One simple way to generate a sequence of this form, as described in the beginning of the section, is to perform all operations one column at a time, i.e.~lowering the first vertical barrier, setting the required gradients to shuttle every qubit adjacent to that vertical barrier and then move on to the second vertical barrier and so on. This yields a sequence of parallel shuttling operations of depth $N$. This solution is always possible for any flow matrix $F$. However, as can be seen in \cref{subsec:configurations} for some flow matrices this is far from an optimal solution. Below we set out in detail an algorithm that finds better (shorter sequences) solutions for many flow matrices. The algorithm is based on the idea that some columns of the flow matrix $F$ can be `dependent' on each other. For instance two columns could be composed of the exact same operations (up to a shift accounting for the fact that the diagonal lines do not run along the rows but diagonally). This means we can perform the shuttle operations in the two columns simultaneously by lowering barriers corresponding to these columns and setting the required gradient. More complicated forms of dependence are also possible. We can use dependence of columns to perform operations in parallel. For instance if a command \code{HS[L]} calls for exactly the same shuttling events to happen on two columns (up to a constant vertical shift proportional to the horizontal distance of the two columns) we can perform these shuttling operations in a single time-step. \\

 This notion of (in)dependence of columns is captured by a call to an `independence subroutine'. We call these subroutines $\mathbf{CheckIndependence}(S,v)$ which takes in a set of columns $S$ of the flow matrix $F$ of and a column $v$ of the flow matrix $F$ and decides whether $v$ is independent of the elements of $S$ and $\mathbf{DependenceSet}(S,v)$ which takes in a set of columns $S$ and a column $v$ and returns a subset $A$ of $S$ containing all the columns on which $v$ depends. We will discuss various versions of these subroutines leading to more or less refined notions of independence (and thus longer of shorter shuttling sequences) in \cref{appsec:Shuttling algorithm}. We list all subroutines discussed in \cref{appsec:Shuttling algorithm} in \cref{tab:subroutines} together with their relative power and time complexity. Here we just treat the subroutines as a given and build the algorithm around it. This algorithm does not always yield optimal sequences of parallel shuttling operations, but it can be run using a polynomial amount of classical side-resources given that the subroutine can be constructed efficiently, (see \cref{thm:main algo complexity}) while we expect an algorithm that always produces optimal shuttling sequences to require exponential computational resources. Below we give a pseudo-code version of the algorithm. Not that this algorithm only produces sequences of parallel shuttling operations where the ordering of the operations does not matter. See \cref{appsec:Shuttling algorithm} for more details on how this property is guaranteed.

 \onecolumngrid 

\begin{algorithm}[H]
\caption{Generate list of parallel shuttle operations}
\label{alg:shuttling algorithm}
\begin{algorithmic}[1]

\Require{Flow matrix $F$}
\Ensure{List of shuttle operations $L$}\\

\State // We will consistently write columns of the flow matrix $F$ as $v_i$ where $i$ indicates
\State // the column index of $v_i$ in $F$. \\

\State \textbf{Set} $S$ to an empty list\\

\State // Below we construct a set of independent columns $S$ and sets of dependence $A_i$ for the dependent columns $v_i$.\\ 

\For{$i \in [0:N-2]$}

	\State \textbf{Set} $v_i$  to the $i$'th column of $F$ \\

	\State //Check if the column $v_i$ is independent of the columns already in the set $S$. This requires a  	\State   // subroutine call to $\mathbf{CheckIndependence}$. See Appendix for the construction of this subroutine.\\

	\If{$\mathbf{CheckIndependence}(v_i,S)$ is TRUE}\\

		\State // The function $\theta$ maps the symbols $*,~re,~le$ to $e$. We must do this since we want to make an operation \State // out of $v_i$ later and the wildcard elements $*,~re,~le$ do not strictly correspond to operations. Other \State // choices are possible here but in keeping with the idea of doing a \State //  minimal amount of operations, the mapping to $e$ is a good choice.\\

		\State \textbf{Add} $\theta(v_i)$ to $S$

		\State \textbf{Set} $A_i$ to $\{v_i\}$
		
	\Else

		\State \textbf{Set} $A_i$ to  $\mathbf{DependenceSet}(S,v_i)$

	\EndIf

\EndFor\\

\State // Initialize an empty ordered set that will contain all \code{HS[L]} commands in sequence.\\

\State \textbf{Set} $L$ to an empty ordered set

\For{$v_i \in S$}\\

	\State // Initialize an empty set that will contain all tuples for a single \code{HS[L]} command.\\

	\State \textbf{Set} \code{L} to an empty set

	\For{$j \in [0:N-2]$}\\

		\State // Check if $v_i$ is in the dependence set $A_j$.\\

		\If{$v_i \in A_j$}

		\State // Loop over all components of $v_i$.

			\For{$k \in [0:\mathrm{length}(v_i)-1]$}\\

				\State // $\phi$ maps the $r,l,e$ valued column $v$ to an $1,-1,0$ valued vector as $\phi(r)=1,\phi(l)=-1,\phi(e)=0$.\\

				\If{$\phi\big[(v_i)_k\big]\neq 0$}

					\State \textbf{Add} $\big(j,k-(i-j),\phi\big[(v_i)_k\big]\big)$ to \code{L}

				\EndIf

			\EndFor

		\EndIf

	\EndFor

	\State \textbf{Add} \code{HS[L]} to $L$

\EndFor
\State\Return{$L$}
\end{algorithmic}
\end{algorithm}

\newpage
\twocolumngrid

\begin{theorem}\label{thm:main algo complexity}
The algorithm described in \Cref{alg:shuttling algorithm} has a time complexity upper bounded by 
\begin{equation}\label{eq:complexity}
\begin{aligned}
&O(N^4) + N \mkern1mu{\cdot}\mkern1mu O\big(\mathbf{CheckIndependence}(S,v_i)\big) \\ &\hspace{16mm}+ N\mkern1mu{\cdot}\mkern1mu O\big(\mathbf{DependenceSet}(S,v_i)\big),
\end{aligned}
\end{equation}
where $N$ is the number of columns in the input flow matrix $F$.\\
\noindent The subroutines $\mathbf{CheckIndependence}(S,v_i)$ and $\mathbf{DependenceSet}(S,v_i)$ both take in a set $S$ of independent columns of the flow matrix $F$ and a column $v_i$ of the flow matrix $F$ and respectively check whether $v$ is independent of the set $S$ or produce a subset $A$ of $S$ on which $v$ depends. Various versions of these subroutines are discussed in \cref{appsec:Shuttling algorithm} and their time complexities are given in \cref{tab:subroutines}.
\end{theorem}
\begin{proof}
Begin by noting that the \cref{alg:shuttling algorithm} consists of two independent $\mathbf{For}$-loops. The first $\mathbf{For}$-loop (lines 2-11) calls its body $N$ times (ignoring constant factors). Calling the $\mathbf{For}$-loop body (lines 3-10) in the worst case requires calling both $\mathbf{CheckIndependence(~)}$ and $\mathbf{DependenceSet(~)}$ plus some constant time instructions. This means the first $\mathbf{For}$ loop has a worst case complexity of $N\cdot O\big(\mathbf{CheckIndependence(~)}\big) + N\cdot O\big(\mathbf{DependenceSet(~)}\big)$.\\

\noindent The second $\mathbf{For}$-loop (lines 13-25) consists of three nested $\mathbf{For}$ loops of length $O(N)$ with an $\mathbf{If}$-clause inside the first two $\mathbf{For}$-loops (line 16) constant time operation at the bottom (line 19). The first $\mathbf{For}$-loop can be seen to be of order $O(N)$ by noting that the set of independent columns $S$ can be no bigger than $N$ in which case all columns are independent. The second $\mathbf{For}$-loop (line 15) is $O(N)$ bounded by construction. Note that the $\mathbf{If}$ clause on line 16 can take time $O(N)$ to complete since for any dependency set $A_j$ we can only say that $|A_j|\leq N$ (since $A_j$ is a subset of the set of all columns of $F$). The third loop is also $O(N)$ bounded since $\mathrm{length}(v_i)\leq N$ for all columns $v_i$ of $F$. Tallying up all contributions we arrive at \cref{eq:complexity}, which completes the argument.
\end{proof}

 \begin{table*} \small
 \begin{tabular}{|c|l|l|}
 \hline
 Name & Time Complexity & Relative power\\
 \hline
 \multirow{2}{*}{Simple} & $O\big(\mathbf{CheckIndependence(~)}\big) = O(N M)$ & \multirow{2}{*}{Shorter sequences than line-by-line.}\\
 &$O\big(\mathbf{IndependenceSet(~)}\big) = O(N M)$ & \\
 \hline
 \multirow{2}{*}{k-commutative} & $O\big(\mathbf{CheckIndependence(~)}\big) = O(N M M^k k^4)$ & Shorter sequences than `Simple'.\\
 &$O\big(\mathbf{IndependenceSet(~)}\big) = O(N M M^k k^4)$ & Shorter sequences for increasing $k$.\\
 \hline
 \multirow{2}{*}{ Greedy commutative} & $O\big(\mathbf{CheckIndependence(~)}\big) = O(N M^3)$ & Shorter sequences than `Simple'.\\
 &$O\big(\mathbf{IndependenceSet(~)}\big) = O(N M^3)$ &Relation to `k-commutative' unknown. \\
 \hline
 \end{tabular}
 \caption{Table listing the time complexity and relative power of the $\mathbf{CheckIndependence(~)}$ and $\mathbf{IndependenceSet(~)}$ for three different classes of subroutine. The parameters $N$ and $M$ are the size of the QDP grid and the size of the input set $S$ respectively. The subroutine classes `simple' and `greedy commutative' can be run in polynomial time while the class `k-commutative' is fixed-parameter-tractable, with independent parameter $k$. This subroutine yields increasingly better results (shorter shuttling sequences) for increasing $k$ but the time complexity grows rapidly with $k$. See \cref{appsec:Shuttling algorithm} for a detailed description of these subroutines. For an illustration of the advantages of these algorithms, one can consider the shuttle commands given in \cref{subsec:configurations}. A naive line-by-line approach will take $N$ timesteps while it is easy to see that the above algorithms find sequences of length one.}\label{tab:subroutines}

 \end{table*}

This concludes our discussion of parallel shuttling operations. Before we move on however, it is worth pointing out an interesting example where this shuttling can be used a subroutine to perform more complicated operations. This example will also be of use later when discussing parallel measurement in \cref{subsec:parallel measurements} and the mapping of quantum error correction codes in \cref{sec:error correction codes,sec:surface code implementation,sec:color code implementation}.

\subsubsection{Selective parallel single-qubit rotations}\label{subsubsec:selective parallel}

In this section we will discuss a particular example that illustrates the use of abstracting away the complexity of parallel shuttling. Imagine a QDP grid initialized in the so called \emph{idle} configuration. This configuration can be seen in \cref{fig:boardstate configurations}. We will focus on the qubit in the odd columns (i.e.~the set $\mc{B}$). Imagine a subset $S$ of these qubits to be in the state $\ket{1}$ and the remainder of these qubits to be in the state $\ket{0}$. The qubits on in the set $\mc{R}$ can be in some arbitrary (and possibly entangled) multiqubit state $\ket{\Psi}$. We would like to change the states states of the qubits in the set $S$ to $\ket{0}$ without changing the state of any other qubit. Due to the limited single qubit gates (see \cref{subsubsec:single qubit rotations}) available in the QDP this is a non-trivial problem for some arbitrary set $S$. However using the power of parallel shuttling we can perform this task as follows. Begin by defining the set $\hat{S}$ to be the complement of $S$ in $\mc{R}$. Now we begin by performing the parallel shuttling operation
\begin{equation}
\code{HS[L]},\;\;\;\; \code{L} = \{ (i,j,1) \;\;\|\;\; (i,j) \in \hat{S}\}.
\end{equation}
Here we abuse notation a bit by referring to $\hat{S}$ as the set of locations of the qubits in $\hat{S}$.
This operation in effect moves all qubits in $\hat{S}$ out of $\mc{R}$ (and into $\mc{B}$, note that the dots the qubits are being shuttled in are always empty because of the definition of the idle configuration). Now we can use a semi-global single qubit rotation (as discussed in \cref{subsubsec:single qubit rotations}) to perform an $X$-rotation on all qubits in $\mc{R}$, which is now just all qubits in the set $S$. This flips changes the states of the qubits in $S$ from $\ket{1}$ to $\ket{0}$ without changing the state of any other qubit. Following this we can restore the \boardstate to its original configuration by applying the parallel shuttling command
\begin{equation}
\code{HS[L]},\;\;\;\; \code{L} = \{ (i,j,-1) \;\;\|\;\; (i,j) \in \hat{S}\}.
\end{equation}
Now we have applied the required operation. Note that at no point we had to reason about the structure of the set $S$ itself. This complexity was taken care of by the classical subroutines embedded in \code{HS[L]}. Next we discuss performing parallel two-qubit gates.

\subsection{Parallel two-qubit gates}\label{subsec:parallel two-qubit gates}

Similar to parallel shuttling it is in general rather involved to perform parallel two-qubit operations in the QDP. We can again define parallel versions of the OPCODES for two-qubit operations and then analyze how to perform them as parallel as possible (again having access to classical side computation).
\begin{figure}[ht]\label{fig:parallel grid interactions}
\centering
\begin{tabular}{|c|c|}
\hline
OPCODE & Effect \\
\hline
\code{HI[L]} & Perform \code{VI[(i,j)]} for \code{$(i,j)\in$ L} \\
\code{VI[L]} & Perform \code{HI[(i,j)]} for \code{$(i,j)\in$ L} \\
\hline
\end{tabular}
\end{figure}

Given an \boardstate and a \code{HI[L]} command one could use an algorithm similar to the algorithm presented for shuttling. We can again construct a matrix $F$  such that $F_{ij}=1$ is for all tuples $(i,j)$ in \code{L} indicating the locations where we desire a two-qubit operation to happen and $F_{ij}=0$ everywhere else. Now we can use the algorithm presented above for shuttling to decompose the matrix $F$ into a series of parallel \code{HI[L]} operations. However, since we have $\mathrm{CHPASE}^2=\id$ the independence subroutine reduces to linear independence of the columns of $F$ modulo $2$. This means we can find an \emph{optimal} decomposition into parallel operations by finding the Schmidt-normal~\cite[Chapter~14]{gorodentsev2016algebra} form of the matrix $F$ (Note that we do have to `tilt' the matrix $F$ to account for the fact that as posed the diagonal lines of the matrix $F$ are its `rows'). We can make the same argument given a \boardstate and a \code{VI[L]} command but now the Schmidt-normal form must be found modulo $4$ as $(\sqrt{\mathrm{SWAP}})^4 =\id$. As both addition modulo $2$ ($\mathbb{Z}_2$) and addition modulo $4$ ($\mathbb{Z}_4$) are principal ideal domains both of the Schmidt-normal forms can be found efficiently and generate optimal sequences of parallel two-qubit interactions. The depth of the sequence of operations is now proportional to the rank of the matrix $F$ over $\mathbb{Z}_2$ (\cphase) or $\mathbb{Z}_4$ (\sw). However, as mentioned before, the parallel operation of two-qubit gates in the QDP will mean taking a hit in operation fidelity vis-a-vis the more controllable line-by-line operation~\cite{veldhorst2017crossbar}. Since this operation fidelity is typically a much larger error source than the waiting-time-induced decoherence stemming from line-by line operation we will for the remainder of the paper assume line-by-line operation of the two-qubit gates. This will have an impact when performing quantum error correction on the QDP which we will discuss in more detail in \cref{sec:discussion}.\\

\noindent For the sake of completeness we also define a parallel version of the $\mathrm{CNOT}$ OPCODE. The same considerations of parallel operation hold for the parallel use of $\mathrm{CNOT}$ gates as they hold for the \cphase and \sw gates. We continue the discussion of parallelism in the QDP by analyzing parallel measurements.

\begin{figure}[ht]\label{fig: grid cnot}
\centering
\begin{tabular}{|c|c|}
\hline
OPCODE & Effect \\
\hline
\code{VC[L]} & Perform \code{VC[(i,j)]} for every $(i,j)$ in \code{L}\\
\hline
\end{tabular}
\end{figure}

\subsection{Parallel Measurements}\label{subsec:parallel measurements}

Performing measurements on an arbitrary subset of qubits on the QDP is in general quite involved. Every qubit to be measured requires an ancilla qubit and this ancilla qubit must be in a known computational basis state, and an empty dot must be adjacent as a reference for the readout process. The qubits must then be shuttled such that they are horizontally adjacent to their respective ancilla qubits and must also be located in such a way such that they are in the right columns for the PSB process to take place (revisit \cref{subsubsec:measurement} for more information). This can be done using the algorithm for parallel shuttling presented above but in the worst case this will take a sequence of depth $O(N)$ parallel shuttle operations. On top of the required shuttling the PSB process itself (from a control perspective similar to shuttling) must be performed in a way that depends on the \boardstate and the configuration of the qubit/ancilla pairs. In general this PSB process will be performed line-by-line (for the fidelity reasons mentioned in the beginning of the section) and hence requires a sequence of depth $O(N)$ parallel grid operations (plus the amount of shuttling operations needed to attain the right measurement configuration in the first place). Due to this complexity we will not analyze parallel measurement in detail but rather focus on a particular case relevant to the mapping of the surface code. But first we define a parallel measurement OPCODE \code{M[L]} which takes in a list of tuples $(i,j,k)$ denoting locations of qubits to be measured $(i,j)$ and whether the ancilla qubit is to the left $(k=-1)$ or to the right $(k=1)$ of the qubit to be measured\\
\begin{figure}[ht]\label{fig: parallel measurement}
\centering
\begin{tabular}{|c|c|}
\hline
OPCODE & Effect \\
\hline
\code{M[L]} & Perform \code{M[$(i,j,k)$]} for every $(i,j,k)$ in \code{L}\\
\hline
\end{tabular}
\end{figure}

\subsubsection{A specific parallel measurement example}

Let us consider a specific example of a parallel measurement procedure that will be used in our discussion of error correction. We begin by imagining the \boardstate to be in the \emph{idle} configuration (\cref{fig:boardstate configurations} top left). We next perform the shuttle operations needed to change the \boardstate to the \emph{measurement} configuration. This configuration (and how to reach it by shuttling operations from the idle configuration) will be discussed \cref{subsec:configurations} and can be seen in \cref{fig:boardstate configurations} (c). Next take the qubits to be measured in the parallel measurement operation to be the red qubits in \cref{fig:boardstate configurations}. The qubits directly to the right or to the left of those qubits will be the required readout ancillas (blue in \cref{fig:boardstate configurations}). We will assume that the readout ancillas are in the $\ket{0}$ state. If some ancilla qubits are in the $\ket{1}$ state instead we can always perform the procedure given in \cref{subsubsec:single qubit rotations} to rotate them to $\ket{0}$ without changing the state of the other qubits on the grid. Note that all the ancilla qubits are in the set $\mc{B}$ whereas the qubits to be read out are in the set $\mc{R}$. This means that we can perform the PSB process by attempting to shuttle the qubit to be measured (red) into the sites occupied by the ancilla qubits (blue). In principle we could perform this operations in parallel by executing the operations

\begin{align}
&\code{VS[L]},\;\;\;\;  \code{L} = \{ (i,j,1)\;\|\;i = 0\mod 2, \notag\\
& \hspace{20mm}j=1 \mod 2, \, i+j = 1 \mod 4 \}
\end{align}

to bring the qubits to be measured (red) horizontally adjacent to the ancilla qubits (blue) and then 
\begin{align}
&\code{M[L]}, \notag \\  &\code{L} = \{ (i,j,1)\;\|\; i= 1 \mod 4,\,   j = 1\mod 4\} \label{eq:par meas 1}
\end{align}

and
\begin{align}
&\code{M[L]},\notag \\ &\code{L} = \{ (i,j,-1)\;\|\; i= 3 \mod 4,\,
  j = 3\mod 4\}. \label{eq:par meas 2}
\end{align}

All of these operations can be performed in a single time-step. However for fidelity and control reasons laid out in the beginning we would prefer to perform these operations in a line-by-line manner. In particular we would like to perform these operations one row at a time since this gives us the ability to control both diagonal and vertical lines individually for each measurement. However we must take care to avoid spurious operations. For instance when performing measurements on the qubits at locations $(1,1)$ and $(1,5)$ we must avoid also performing a measurement on the qubit at location $(5,5)$. To avoid this situation we will bring only the bottom row of qubits to be measured horizontally adjacent to the ancilla qubits, perform the PSB process and readout on that row only and then shuttle the qubits to be measured back down again. This we repeat going up in rows until we reach the end of the grid. More formally we perform the following sequence of operations. 
\small
\begin{algorithm}[H]
\caption{Loop over OPCODES to perform line-by-line measurements}
\label{alg:line by line measurement}
\begin{algorithmic}[1]
\For{$i \in [0:N-2]$}

	\If{$i= 1 \mod 4 $}
		\State $\code{VS[L]},\;\;\;\; \code{L} = \{(i-1,j,-1)\;\|\; j =1 \mod 4\}$

		\State $\code{M[L]},\;\;\;\;\;\: \code{L} = \{ (i,j,1)\;\|\; j = 1\mod 4\}$

		\State $\code{VS[L]},\;\;\;\; \code{L} = \{(i-1,j,1)\;\|\; j =1 \mod 4\}$

	\EndIf

	\If{$i= 3 \mod 4 $}

		\State $\code{VS[L]},\;\;\;\; \code{L} = \{(i-1,j,-1)\;\|\; j=3 \mod 4\}$

		\State $\code{M[L]},\;\;\;\;\;\: \code{L} = \{ (i,j,-1)\;\|\; j = 3\mod 4\}$

		\State $\code{VS[L]},\;\;\;\;\code{L} = \{(i-1,j,1)\;\|\; j =3 \mod 4\}$

	\EndIf

\EndFor
\end{algorithmic}	
\end{algorithm}
\noindent We will use this particular procedure when performing the readout step in a surface code error correction cycle in \cref{sec:surface code implementation}. This concludes our discussion of parallel operation on the QDP. We now move on to highlight some \boardstate configurations that will feature prominently in the surface and color code mappings.

\subsection{Some useful grid configurations}\label{subsec:configurations}

\begin{figure*}
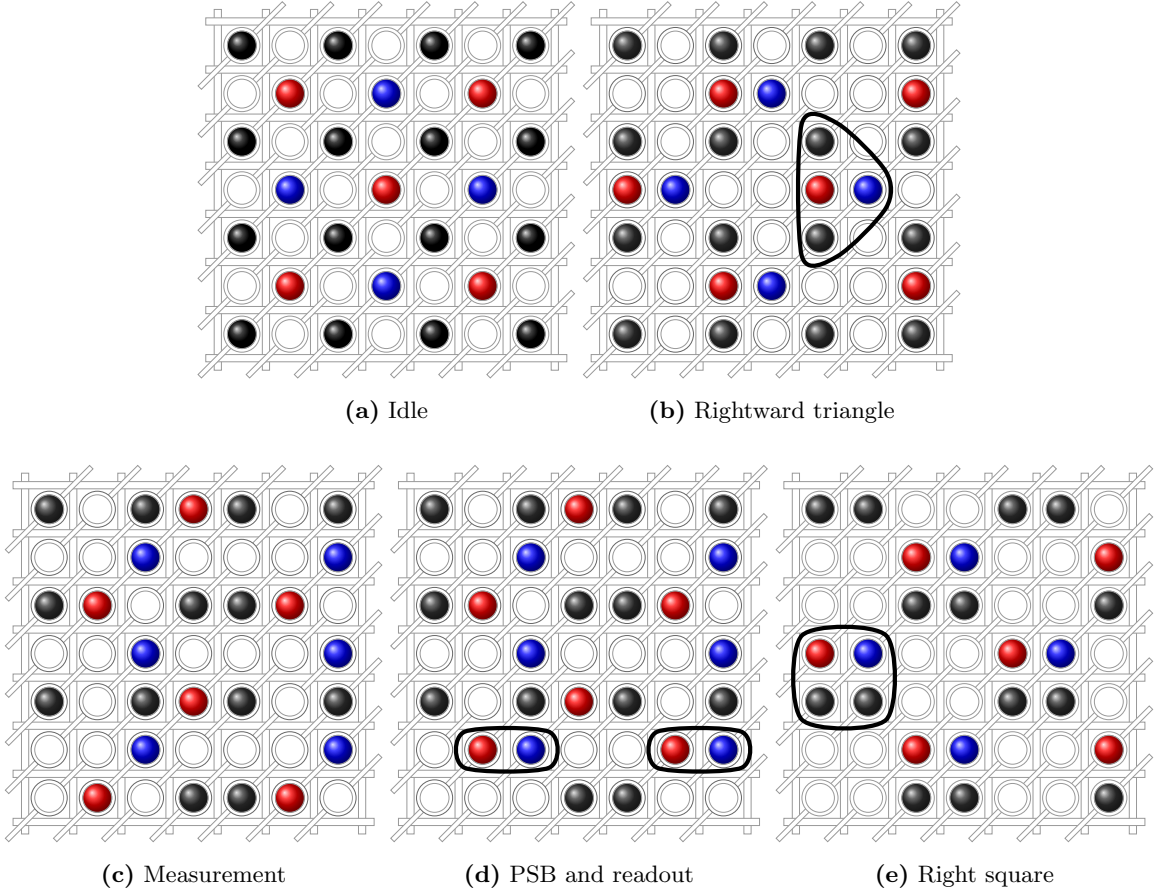

\include{configs}
\caption{Useful \boardstate configurations. We denote data qubits with black color, $X$-measurement qubits by red and $Z$-measurement qubits by blue. Those will collect the parity of the data qubits in one error correction cycle, and one is the others reference at the PSB measurement.   \textbf{(a)} The idle  configuration is a starting point of all algorithms.  All qubits are spred out and well separated. \textbf{(b)} The triangle configurations (here we have a rightward triangle, see the frame in the figure) is assumed when the proximity of measurement qubits to data qubits is required. This is the case for the parity measurements in error correction cycles.  \textbf{(c)} The measurement configuration is formed to bring $X$- and $Z$-measurement qubits close to each other, such that a row can be selected in which the measurement is performed. \textbf{(d)}  Certain measurement qubits are brought to adjacent dots in order to perform the PSB-based measurement and  readout in a  line-by-line fashion (encircled qubits). Since the rest of the grid is in the measurement configuration, individual control over the barrier lines and one potential is guaranteed without spurious measurements.    \textbf{(e)} The (right) square configuration is a mid-way point between the idle and (right) triangle configuration. Going through the square configuration keeps the shuttling algorithm managable, as not more that 2 different heights of the dot potentials are employed. One of the characteristic squares is enframed in the figure.  }\label{fig:boardstate configurations}
\end{figure*}

There are several configurations of the \boardstate that show up frequently enough (for instance in the error correction codes in \cref{sec:surface code implementation}) to merit some special attention. In this section we list these specific configurations and show how to construct them.
\subsubsection{Idle configuration}
The idle configuration is the configuration in which the QDP is initialized. As shown in \cref{fig:boardstate configurations} it has a checkerboard pattern of filled and unfilled sites. In this configuration no two-qubit gates can be applied between any qubit pair but since it minimizes unwanted crosstalk between qubits \cite{veldhorst2017crossbar}, it is good practice to bring the system back to this configuration when not performing any operations. For this reason we consider the idle configuration to be the starting point for the construction of all other configurations.
\subsubsection{Square configuration}
As seen in \cref{fig:boardstate configurations}(e) the square configurations consist of alternating filled and unfilled $2\times 2$ blocks of sites. The so-called right square configuration can be reached from the idle configuration by a shuttling operation \code{HS[L]} with the set \code{L} being 
\begin{align}
\code{L} &= \{(i,j,1) \;\|\;i=1\mod 2, \notag \\
&\hspace{12mm}j=1\mod 2 , i+j=2 \mod 4\}\notag\\
&\hspace{5mm}\cup\{(i,j,-1) \;\|\;i=0\mod 2,\notag\\
&\hspace{12mm}j=1\mod 2, i+j=3 \mod 4\}.
\end{align}
Note that this operation only takes a single time-step, the square configuration is shown in \cref{fig:boardstate configurations}(e).
The right square configuration is characterized by the red ($Z$-) ancilla being in the left corner of every square.  Another flavor of this configuration is the left square configuration, where the red ancilla is in the upper right corner, and the blue one in the left.
The left square configuration can be reached from the idle configuration by a shuttling operation \code{HS[L]} with the set \code{L} being  .

\begin{align}
\code{L} &= \{(i,j,1) \;\|\;i=0\mod 2,\notag \\
&\hspace{12mm}j=0\mod 2 ,\, i+j=2 \mod 4\} \notag \\
&\hspace{5mm}\cup\{(i,j,-1) \;\|\;i=1\mod 2, \notag\\
&\hspace{12mm}j=0\mod 2,\, i+j=1 \mod 4\}.
\end{align}

These configurations are used as an intermediate step for us to reach the triangle configurations.
\subsubsection{Measurement Configuration}
The measurement configuration can be reached from the idle configuration in three time-steps by the following sequence of parallel shuttling operations.
\begin{align}
&\code{HS[A]},\;\;\; \code{A} = \{ (i,j,-1), \;(i-1,j-1,1)\;\notag\\
&\hspace{25mm}\|\;i =1\mod 4, \;j=2 \mod 4\},\notag\\
&\code{HS[B]},\;\;\; \code{B} = \{(i-1,j-1,1)\;\notag\\
&\hspace{25mm}\|\;i =3\mod 4, \;j=1 \mod 4\},\notag\\
&\code{VS[C]},\;\;\; \code{C} = \{ (i,j,-1)\;\|\;i =0\mod 2,\notag\\
&\hspace{23mm} j=1 \mod 2,\;i+j = 1 \mod 4 \}.
\end{align}
This configuration can be seen in \cref{fig:boardstate configurations}(d) and it is an intermediate state in the measurement process of the blue qubits using the red qubits as ancillas. How this measurement protocol works in detail is described in \cref{subsec:parallel measurements}.

\subsubsection{Triangle configurations}
In order to collect the parity of the data qubits in the error correction cycles, we need to align the ancilla qubits with the data qubits, according to the two-qubit gates used. This is reflected in the use of triangle configurations. 
There are two triangle configurations that can be reached in a single parallel shuttling step from the right square configuration. The first one, seen in \cref{fig:boardstate configurations}(b), is called the rightward triangle configuration. It can be reached from the square configuration by the grid operation \code{HS[L]} with the set \code{L} being
\begin{align}
\code{L} &= \{(i,j,-1) \;\|\;0=1\mod 2, \notag\\&\hspace{10mm} \, j=1\mod 2, i+j =3\mod 4\},
\end{align}
which does as much as to shuttle the right data qubit of every square (enframed squares in \cref{fig:boardstate configurations}(e)) to the empty dot on its right. 
In this configuration, we are able to perform high-fidelity two-qubit gates between the two data qubits and the ancilla in every triangle. In order to reach the neighboring pair of data qubits with the same ancilla, we start from the left square configuration and shuttle the left data qubit to the left. Operationally, we would do \code{HS[L]} with
\begin{align}
\code{L} &= \{(i,j,1) \;\;\|\;\;i=0\mod 2, \notag\\ &\hspace{10mm} j=0\mod 2, \, i+j =2\mod 4\}.
\end{align}
Note again that these parallel shuttling operations can be performed in a single time step. From these configurations the idle configuration can also be reached in a single time step.
In the next section these configurations will feature prominently in the mapping of several quantum error correction codes to the QDP architecture.

\section{Error correction codes}\label{sec:error correction codes}
In this section we will apply the techniques we developed in the previous sections to map several quantum error correction codes to the QDP. 
\subsection{Introduction}
First we recall some basic facts about quantum error correction codes and topological stabilizer codes in particular. The focus will be on practical application, for a more in depth treatment of quantum error correction and topological error correction codes we refer to~\cite{lidar2013quantum}. Recall first the Pauli operators on a single qubit:
\begin{equation}
X = \begin{pmatrix} 0 & 1 \\ 1 & 0 \end{pmatrix},\;\;\;\;Z= \begin{pmatrix} 1 & 0 \\ 0 & -1 \end{pmatrix}.
\end{equation}
Given a system of $n$ qubits we denote by $P_i$ the Pauli operator $P\in \{X,Z\}$ acting on the $i$'th qubit. With this definition we can see write the $n$ qubit Pauli group $\mc{P}_n$ as the group generated by the operators $\{X_i,Z_j \;:\;i,j\in [1:n]\}$ under matrix multiplication. A stabilizer quantum error correction code acting on $n$ physical qubits and encoding $k$ logical qubits can then be defined as the joint positive eigenspace of an abelian subgroup $\mc{S}$ of $\mc{P}_n$ generated by $n-k$ independent commuting Pauli operators. Operationally, this code is then defined by measuring the generators of $\mc{S}$ and if necessary perform corrections to bring the state of the system back into the positive joint eigenspace of these generators. This is a very general definition and it is not guaranteed that a code defined this way yields any protection against errors happening. Below we will see some common examples of stabilizer error correction codes that do have good protection against errors. On top of that, these codes have the desirable property that their stabilizers are in some sense `local'. That is they can be implemented on qubits lying on a lattice such that the stabilizer generators can be measured by entangling a patch of qubits that is small with respect to the total lattice size. The most well known example of a code of this type is the so-called planar surface code.

\subsection{Planar surface code}
The planar surface code is probably the most well known practical quantum error correction code due to its high threshold~\cite{wang2011surface}, the availability of efficient decoding algorithms~\cite{fowler2009high}. To construct the planar surface code (in particular we will use the so-called rotated planar surface code~\cite{horsman2012surface}, as it uses less physical qubits per logical qubit) we will consider a regular $n\times n$ square lattice of degree four (every node has four connected neighbors) and we will place qubits on each node. We will define the generators of the abelian group $\mc{C}$ that defines the surface code by alternately placing $X$- and $Z$-quartets on the faces of the lattice (in \cref{fig:surface code} the red faces correspond to $X$-stabilizer quarters while the green faces correspond to $Z$ stabilizer quartets). This $X (Z)$ will indicate that we pick the generator $X\tn{4}$ ($Z\tn{4}$) on the four qubits on the corners of the $X$ ($Z$) face. Note that this means that all of the generators commute with each other since they either act on disjoint sets of qubits or act on sets that have an overlap of exactly $2$ qubits. Since $XZ = -ZX$ we have that $X\tn{2}Z\tn{2} = Z\tn{2}X\tn{2}$ which means that all generators commute. These generators (plus appropriate generators on the boundary of the lattice) define a stabilizer group which specifies a code space of dimension $2$, i.e.~a single logical qubit. We can locally measure these $X (Z)$ stabilizers by using the circuits~\cite{terhal2015quantum,gottesman1998theory,lidar2013quantum,fowler2012surface} illustrated in \cref{fig:stabilizer circuits}. This construction calls for one ancilla qubit per lattice face.

\begin{figure*}[t]
 \centering
 \begin{minipage}{.48\textwidth}
 \centering
 \includegraphics[scale=0.4]{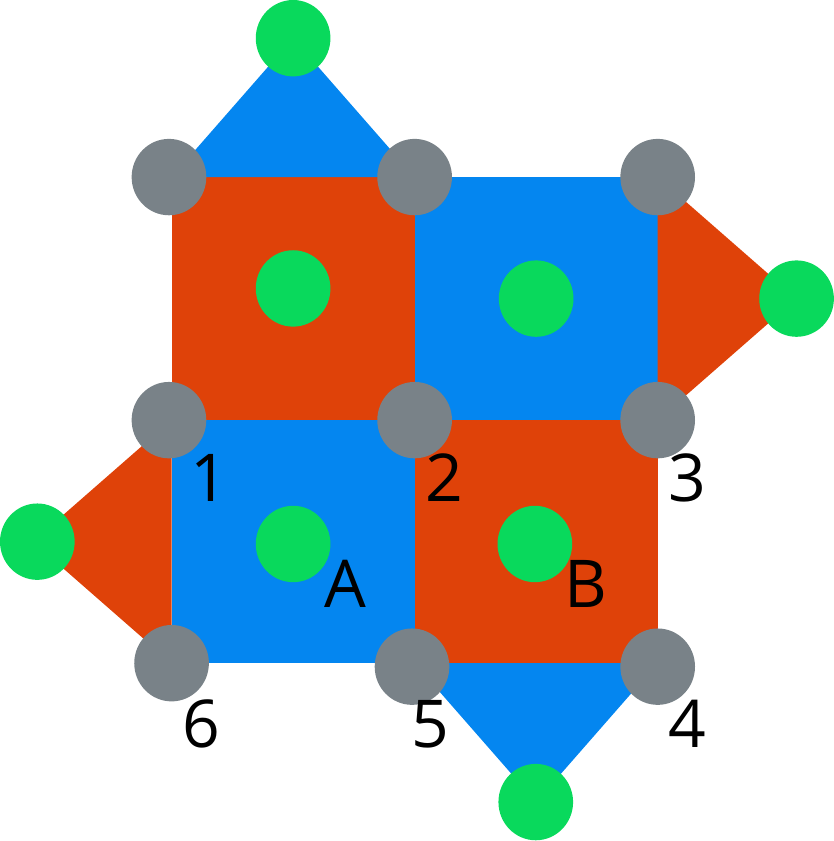}
 \caption{Schematic representation of a distance three rotated planar surface code~\cite{horsman2012surface}. The gray circles represent the data qubits supporting the code. The green circles represent ancilla qubits, which are used to perform the stabilizer measurements which define the code. These stabilizer measurements are represented by the red ( $Z$-type stabilizers) and blue faces ($X$-type stabilizers). The ancilla qubit in the middle of a face will be used to perform a stabilizer measurement of the data qubits on the corners of that face. The actual quantum circuits used to perform these stabilizer measurements are shown in \cref{fig:stabilizer circuits}.}\label{fig:surface code}
\end{minipage}%
$\;\;$
\begin{minipage}{0.48\textwidth}
\centering
{\bf $Z$ stabilizer sequence}
\begin{equation*}
\Qcircuit @C=1em @R=.7em {
	\lstick{\ket{0}_A} 	& \targ 	& \targ 	& \targ 	& \targ 	& \measuretab{M_{Z}}\\
	\lstick{\ket{\psi}_1}	& \ctrl{-1} & \qw		& \qw		& \qw 		& \qw				\\
	\lstick{\ket{\psi}_2}	& \qw 		& \ctrl{-2}	& \qw		& \qw 		& \qw				\\
	\lstick{\ket{\psi}_5}	& \qw 		& \qw		& \ctrl{-3}	& \qw 		& \qw				\\
	\lstick{\ket{\psi}_6}	& \qw 		& \qw		& \qw		& \ctrl{-4}	& \qw				
}
\end{equation*}

{\bf $X$-stabilizer sequence}\\
\begin{equation*}
\Qcircuit @C=1em @R=.5em {
	\lstick{\ket{+}_{B}} 	& \ctrl{1} 	& \ctrl{2} 	& \ctrl{3} 	& \ctrl{4} 	& \measuretab{M_{X}}\\
	\lstick{\ket{\psi}_2}	& \targ	 	& \qw		& \qw		& \qw 		& \qw				\\
	\lstick{\ket{\psi}_3}	& \qw 		& \targ		& \qw		& \qw 		& \qw				\\
	\lstick{\ket{\psi}_4}	& \qw 		& \qw		& \targ		& \qw 		& \qw				\\
	\lstick{\ket{\psi}_5}	& \qw 		& \qw		& \qw		& \targ		& \qw				
}
\end{equation*}
\caption{Quantum circuits for performing the $X$- and $Z$-stabilizer measurements of the planar surface code~\cite{terhal2015quantum,gottesman1998theory,lidar2013quantum,fowler2012surface}. The qubits $A$ and $B$ (see \cref{fig:surface code}) are ancilla qubits used to perform stabilizer measurements on the the data qubits on the corners of the faces defining the code. The data qubits associated to the face of qubit $A$ are $\{1,2,5,6\}$ and likewise $
\{2,3,4,5\}$ for qubit $B$.}\label{fig:stabilizer circuits}
\end{minipage}
\end{figure*}

\subsection{2D color codes}
Another important class of planar topological codes are the 2D color codes~\cite{bombin2006topological}. These codes are defined on 3-colorable tilings of the Euclidean plane. Two popular tilings are the so called $6.6.6.$ and $4.8.8.$ tilings corresponding to hexagonal and square-octagonal tilings respectively. To construct the code qubits are places on all vertices of the tiling and $X$- and $Z$-stabilizers are associated to every tile by applying $X$ ($Z$) to every qubit on the corner of the tile. With suitable boundary conditions this construction encodes a single logical qubit with a distance proportional to $\sqrt{n}$ with $n$ the number of physical qubits. See \cref{fig:color codes} for examples of the $6.6.6.$ and $4.8.8.$ color codes of distance five. Note that these pictures do not include ancilla qubits for measuring the stabilizers. The planar color codes have lower thresholds than the planar surface code but are more versatile when it comes to fault-tolerant gates. The planar color codes support the full Clifford group as a transversal set, making quantum computation on color codes more efficient than on the surface code. In the next section we will focus on mapping these codes to the QDP using the concepts introduced in \cref{sec:An assembly language for crossbar control}.

\begin{figure*}
\begin{minipage}{0.45\textwidth}
\centering
\includegraphics[scale=0.2]{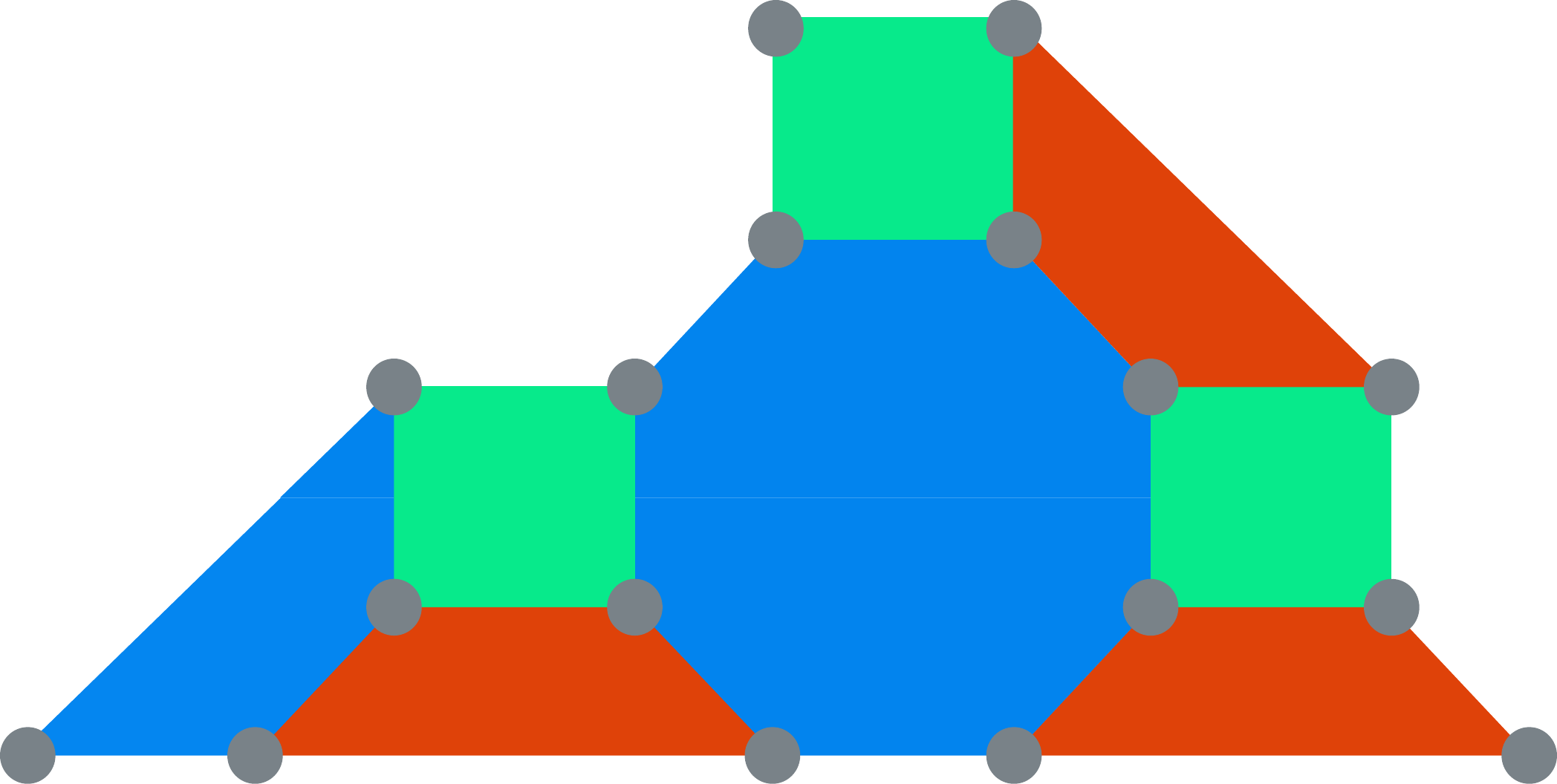}
$\vspace{6mm}$
\includegraphics[scale=0.2]{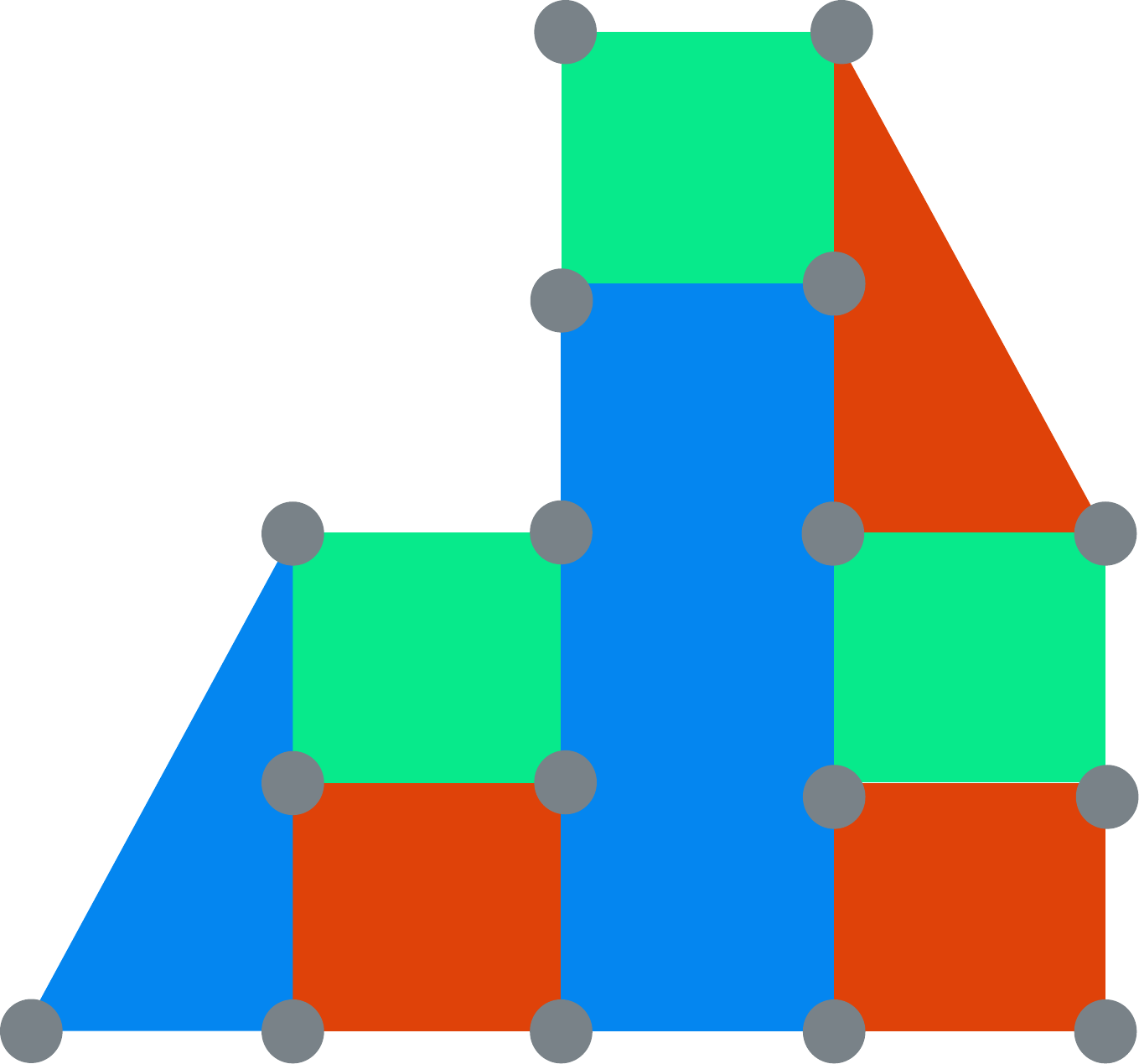}
\end{minipage}
$\;\;\;\;$
\begin{minipage}{0.45\textwidth}
\centering
\includegraphics[scale=0.2]{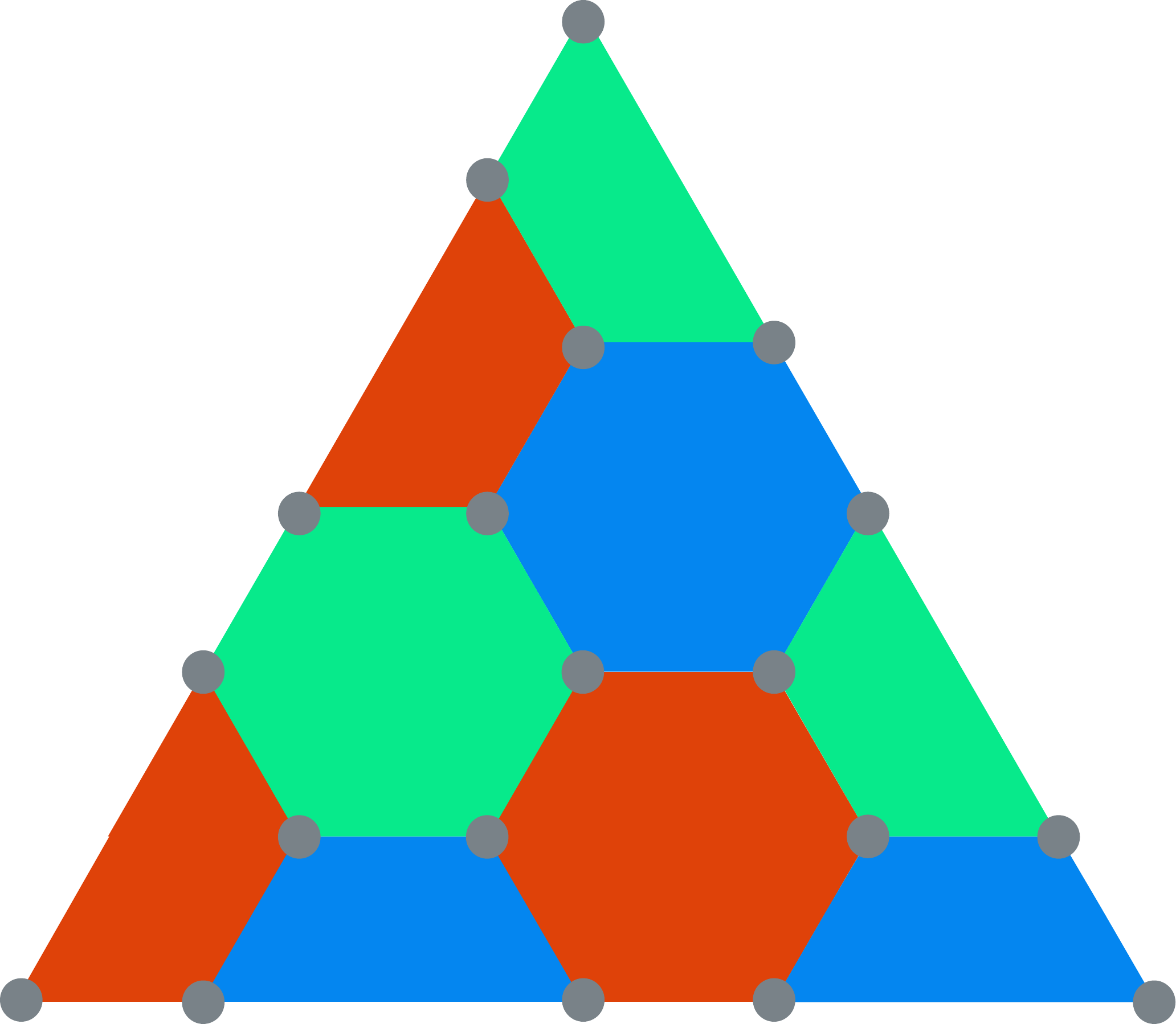}
$\vspace{6mm}$
\includegraphics[scale=0.2]{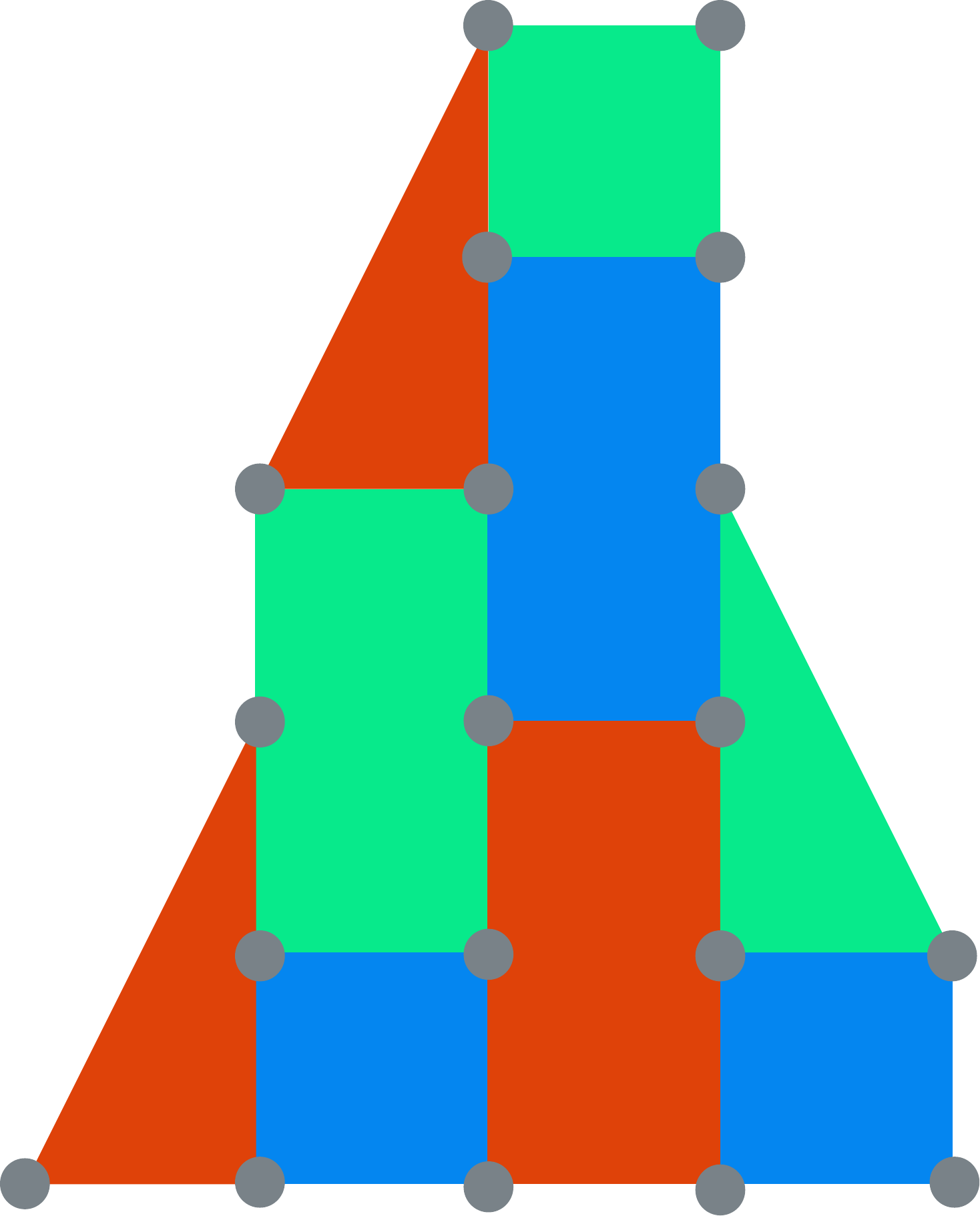}

\end{minipage}
\caption{Distance $5$ examples of the $4.8.8.$ (first from left) and $6.6.6.$ (third from left) color codes~\cite{bombin2006topological} and their deformed versions (second from left and fourth from left respectively). The vertices correspond to data qubits and every colored face corresponds to both an $X$- and a $Z$-stabilizer to be measured. These stabilizers can be measured by using weight $4, 6$ and $8$ versions of the circuits shown in \cref{fig:stabilizer circuits}. The deformation of the codes does not change the code properties at all. They are a visual guide that facilitates the mapping the the crossbar grid in \cref{sec:color code implementation}. }\label{fig:color codes}

\end{figure*}

\twocolumngrid

\subsection{Surface code mapping}\label{sec:surface code implementation}
We now describe a protocol that maps the surface code on the architecture described in \cref{sec:the quantum dot processor}. The surface code layout has a straightforward mapping that places the data qubits on the even numbered columns and the $X$- and $Z$-ancillas on the odd columns. This means we have single-qubit control over all data qubits and all ancilla qubits separately. There are two ways to perform the surface code cycle; we could use either the \sw gate or the \cphase gate as the main two-qubit gate. Since in practice the \sw gate has higher fidelity \cite{veldhorst2017crossbar} we will use this gate. We begin by changing the circuits performing the $X$- and $Z$-stabilizer measurements to work with \sw rather than $\mathrm{CNOT}$. We can emulate a \cnot gate by using two \sw gates interspersed with a $Z$-gate on the control plus some single qubit gates. As described in \cref{subsubsec:CNOT subroutines} the $Z$- and $S$-gates on the ancilla qubit can performed by waiting, which means they can be performed locally while the single qubit operations on the data qubits can be performed in parallel using the global unitary rotations described in \cref{subsubsec:single qubit rotations}. The $X$- and $Z$-circuits using \sw are shown in \cref{fig:swap circuits}. \\

We will split up the quantum error correction cycle by first performing all $X$-type stabilizers (the $X$-cycle) and then all $Z$-type stabilizers ($Z$-cycle). This means we can use the idle $Z$- ($X$-) ancilla to perform a measurement on the $X$- ($Z$-) ancilla at the end of the $X$ ($Z$) cycle. For convenience we included a depiction of the surface code $Z$-cycle unit cell in \cref{fig:unit cells} (right). The qubit labeled `A' is the ancilla used for the $Z$ stabilizer circuit. The numbered qubits are data qubits and the qubit labeled `B' is the qubit used for reading out the `A' qubit. It is also the ancilla qubit for the $X$-cycle. We now describe the steps needed to perform the $Z$-cycle in parallel on the entire surface code sheet. For convenience we ignore the surface code boundary conditions but these can be easily included. The $X$-cycle is equivalent up to different single qubit gates ($XS\ct$ instead of $ZHS\ct$ on the data qubits, $HS\ct$ instead of $S\ct$ on the ancilla) and shifting every operation $2$ steps up, e.g.~setting $i$ to $i+2$. 
\begin{figure*}
\centering
{\bf $Z$ stabilizer sequence}
\begin{equation*}
\Qcircuit @C=0.6em @R=.4em {
	\lstick{\ket{0}_A} 	&\gate{S\ct}	& \ctrlb{0} 	& \gate{Z}	& \ctrlb{0} 	&\gate{S\ct} 	& \ctrlb{0} 	& \gate{Z}	& \ctrlb{0} 	&\gate{S\ct} 	& \ctrlb{0} 	&\gate{Z} 	&\ctrlb{0} 			&\gate{S\ct}		&\ctrlb{0}	&\gate{Z}		&\ctrlb{0} 	&\gate{S\ct}	&\measuretab{Z}\\
	\lstick{\ket{q_1}}	&\gate{ZHS\ct}	& \ctrlb{-1} 	& \qw		& \ctrlb{-1}	& \qw			& \qw			& \qw 		& \qw			&\qw 			& \qw 			&\qw 		&\qw 				&\qw 				&\qw		&\qw 			&\qw 		&\gate{ZHS\ct}	&\qw	\\
	\lstick{\ket{q_2}}	&\gate{ZHS\ct}	& \qw 			& \qw 		& \qw			& \qw			& \ctrlb{-2}	& \qw 		& \ctrlb{-2}	&\qw 			& \qw 			&\qw 		&\qw 				&\qw 				&\qw		&\qw 			&\qw 		&\gate{ZHS\ct}	&\qw	\\
	\lstick{\ket{q_3}}	&\gate{ZHS\ct}	& \qw 			& \qw 		& \qw			& \qw			& \qw			& \qw 		& \qw			&\qw 			& \ctrlb{-3} 	&\qw 		&\ctrlb{-3} 		&\qw 				&\qw		&\qw 			&\qw 		&\gate{ZHS\ct}	&\qw	\\
	\lstick{\ket{q_4}}	&\gate{ZHS\ct}	& \qw 			& \qw 		& \qw			& \qw			& \qw			& \qw		& \qw			&\qw 			& \qw 			&\qw 		&\qw 				&\qw 				&\ctrlb{-4}	&\qw 	 		&\ctrlb{-4}	&\gate{ZHS\ct}	&\qw 
}
\end{equation*}
\caption{$Z$ stabilizer measurement circuit using the \sw as the main two-qubit gate. The $Z$- and $S$-rotations can be performed by the timing procedure described in \cref{subsubsec:single qubit rotations}.}\label{fig:swap circuits}
\end{figure*}
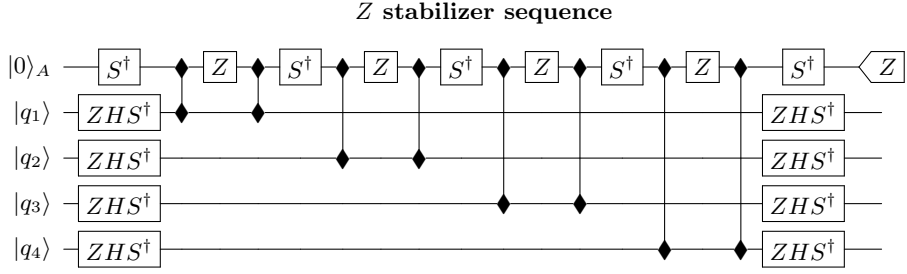

\begin{figure*}
\begin{minipage}{0.30\textwidth}
\centering
\includegraphics[scale=0.87]{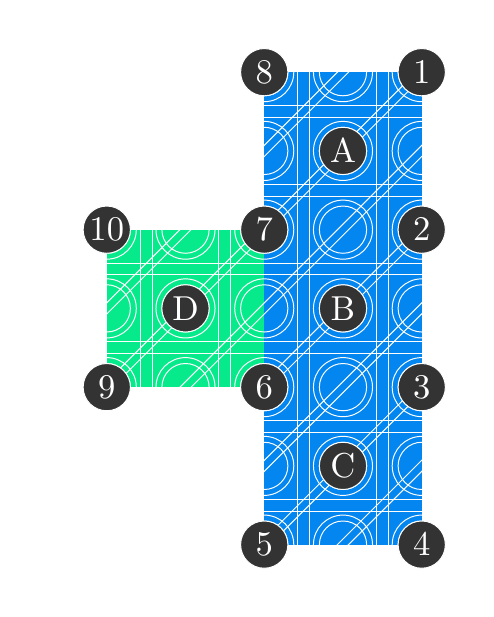}
\end{minipage}
\begin{minipage}{0.30\textwidth}
\centering
\includegraphics[scale=0.87]{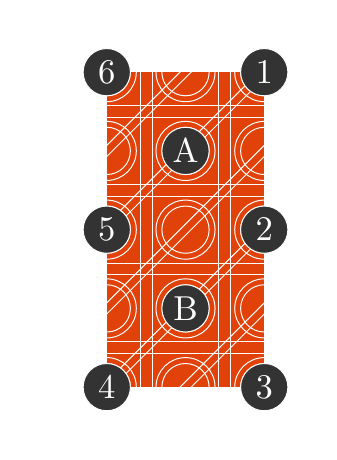}
\end{minipage}
\begin{minipage}{0.30\textwidth}
\centering
\includegraphics[scale=0.87]{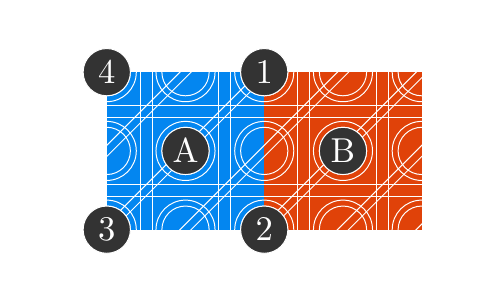}
\end{minipage}
\caption{Unit cells of the deformed $4.8.8.$ and $6.6.6.$ codes (left and middle respectively) and the unit cell of the surface code ($Z$-cycle) with the gray circle corresponding to qubits. For the $4.8.8.$ unit cell the qubit labeled `A' is the ancilla qubit for the octagon (now a rectangle) sub-cell while the qubit labeled `D' is the ancilla for the square sub-cell. The qubit labeled `B' is used to read out the qubit labeled `D' and the qubit labeled `C' is used to read out the ancilla qubit for the octagon cell directly below the square cell (not pictured). The qubits labeled by numbers are the data qubits. For the $6.6.6.$ unit cell the qubit labeled `A' is the ancilla qubit used to perform the stabilizer measurement while the qubit labeled `B' is used to read out the `A' qubit for the unit cell directly to the bottom left (not pictured). The numbered qubits are again data qubits. For the surface code unit cell the qubit labeled `A' is the ancilla used for the $Z$-cycle stabilizer measurement while the qubit labeled `B' is the qubit used to read out the `A' qubit. It is also the qubit used as the ancilla for the $X$-stabilizer cycle. The numbered qubits are again data qubits. Note that this unit cell mirrors when moving upwards. That is, the unit cell above the one pictured will have the ancilla qubit B to the right of qubit A instead of to the left as pictured. }\label{fig:unit cells}
\end{figure*}

\onecolumngrid 

\noindent
\fbox{
\begin{minipage}[t]{\textwidth}
\centering
{\bf The surface code $Z$-cycle}
\begin{enumerate}[{\bf ~Step 1:}]
\setlength\itemsep{0.1em}
	\item Initialize in the idle configuration
	\item Apply $ZHS\ct$ to all qubits in $\mc{R}$ (data) and $S\ct$ to qubits in $\mc{B}$ (ancilla)
	\item Go to right square configuration 
	\item Go to rightward triangle configuration 
	\item Perform \cnot between qubits A and $1$ by performing \code{VC[L]} with
	\begin{equation*}
	\code{L} = \{(i,j)\hspace{2mm} \|\hspace{2mm} i=1\mod 2, j = 0 \mod 2, i+j = 3 \mod 4\} 
	\end{equation*}
	\item Perform \cnot between qubits A and $2$ by performing \code{VC[L]} with
	\begin{equation*}
	\code{L} = \{(i,j)\hspace{2mm} \|\hspace{2mm} i=0\mod 2, j = 0 \mod 2, i+j = 2 \mod 4\} 
	\end{equation*}
	\item Go to idle configuration 
	\item Go to left square configuration 
	\item Go to leftward triangle configuration 
	\item Perform \cnot between qubits A and $3$ by performing \code{VC[L]} with
	\begin{equation*}
	\code{L} = \{(i,j)\hspace{2mm} \|\hspace{2mm} i=1\mod 2, j = 0 \mod 2, i+j = 1 \mod 4\} 
	\end{equation*}
	\item Perform \cnot between qubits A and $4$ by performing \code{VC[L]} with
	\begin{equation*}
	\code{L} = \{(i,j)\hspace{2mm} \|\hspace{2mm} i=0\mod 2, j = 0 \mod 2, i+j = 0 \mod 4\} 
	\end{equation*}
	\item Go to idle configuration 
	\item Apply $ZHS\ct$ to all qubits in $\mc{R}$ (data) and $S\ct$ to qubits in $\mc{B}$ (ancilla)
	\item Apply measurement ancilla correction step for qubit B as described in \cref{subsubsec:selective parallel}
	\item Go to measurement configuration 
	\item Perform Pauli Spin Blockade measurement process as described in \cref{subsec:parallel measurements} using qubit B as ancilla to qubit A 
	\item Go to idle configuration

\end{enumerate}
	\phantom{end}
\end{minipage}
}
\\\phantom{end}\\

\twocolumngrid

\subsection{Color code mapping}\label{sec:color code implementation}
The mapping of the color codes is largely analogous to that of the surface code. We begin with the $6.6.6.$ color code as it is easiest to map. We begin by deforming the tiling on which the color code is defined such that it is more amenable to the square grid structure of the QDP. This is fairly straightforward as can be seen from the $d=5$ example in \cref{fig:color codes}. In the deformed tiling it is clear how to map the code to the crossbar grid layout. We once again place all data qubits in the even columns and all ancilla qubits in the odd columns. This places the unit `hexagon' seen in the deformed code in a $3\times5$ tile on the QDP (see \cref{fig:unit cells} (right) for this unit tile). This places all data qubits in $\mc{R}$ and $2$ extra qubits in $\mc{B}$, both of which could be used as an ancilla in the stabilizer circuit. We will always choose the top qubit (qubit `A') of these two in the hexagon unit cell as the ancilla qubit for the error correction cycles. The extra (bottom) qubit (qubit `B') in the unit cell will be used to perform the readout of the ancilla qubit of the unit hexagon to its direct left. This has the advantage of making the readout process independent of the measurement results of the previous cycles (as was the case in the surface code). Note also that the ancilla qubits are positioned along diagonal lines on the QDP grid. This makes the quantum error correction cycle very analogous to the surface code. We once again must split up the $X$- and $Z$-cycles (again due to the limited single qubit rotations possible). Below we present the steps needed to perform the $Z$-cycle (which now measures a weight $6$ operator). The $X$-cycle is identical up to differing single qubit rotations on the data qubits. \\

\onecolumngrid 

\noindent
\fbox{
\begin{minipage}[t]{\textwidth}
\centering
{\bf The $6.6.6$ color code $Z$-cycle}
\begin{enumerate}[{\bf ~Step 1:}]
\setlength\itemsep{0.1em}
	\item Perform {\bf Steps 1 to 11} in the surface code $Z$-cycle to perform \cnot s between the ancilla (qubit A) and the data qubits $1,2,5,6$ in the unit hexagon and end in the idle configuration
	\item Go to idle configuration but with all even columns up and all odd columns down by performing \code{VS[L]} with
	\begin{align*}
	\code{L} &= \{(i,j,1)\hspace{2mm} \|\hspace{2mm} i=0\mod 2, j = 0 \mod 2\}\\
				&\hspace{10mm} \cup\{(i,j,-1)\hspace{2mm} \|\hspace{2mm} i=1\mod 2, j = 1 \mod 2\}
	\end{align*}
	\item Go to right square configuration 
	\item Go to rightward triangle configuration 
	\item Perform \cnot between qubits A and $3$ by performing \code{VC[L]} with
	\begin{equation*}
	\code{L} = \{(i,j)\hspace{2mm} \|\hspace{2mm} i=1\mod 2, j = 0 \mod 2, i+j = 1 \mod 4\} 
	\end{equation*}
	\item Go to idle configuration 
	\item Go to left square configuration 
	\item Go to leftward triangle 
	\item Perform \cnot by performing between qubits A and $4$ \code{VC[L]} with
	\begin{equation*}
	\code{L} = \{(i,j)\hspace{2mm} \|\hspace{2mm} i=0\mod 2, j = 0 \mod 2, i+j = 2 \mod 4\} 
	\end{equation*}
	\item Go to idle configuration 
	\item Invert {\bf Step 6} by performing \code{VS[L]} with
	\begin{align*}
	\code{L} &= \{(i,j,-1)\hspace{2mm} \|\hspace{2mm} i=0\mod 2, j = 0 \mod 2\}\\
				&\hspace{10mm} \cup\{(i,j,1)\hspace{2mm} \|\hspace{2mm} i=1\mod 2, j = 1 \mod 2\}
	\end{align*}
	\item Apply $ZHS\ct$ to all qubits in $\mc{R}$ (data) and $S\ct$ to qubits in $\mc{B}$ (ancilla)
	\item Go to measurement configuration 
	\item Perform Pauli Spin Blockade measurement process as described in \cref{subsec:parallel measurements} using qubit B as ancilla to read out qubit A (unit cell to the right)
	\item Go to idle configuration

\end{enumerate}
\end{minipage}
}
\\\phantom{end}\\

\twocolumngrid

Next up is the $4.8.8.$ color code. We deform the tiling on which the code is defined similarly to the $6.6.6.$ code. The deformed $4.8.8.$ code lattice can be seen in \cref{fig:unit cells} (left). We again place the data qubits in the set $\mc{R}$ the ancilla qubits in the set $\mc{B}$. See \cref{fig:unit cells} for a layout of the unit cell of the $4.8.8.$ code on the QDP. Note that there are two different types of tiles in this code. The square tile has one qubit (qubit `D' in \cref{fig:unit cells}) in $\mc{B}$, which we will use as ancilla qubit for that tile. The deformed octagon tile has three qubits in $\mc{B}$. We will use the topmost qubit (qubit `A') as the ancilla qubit for the tile while the middle one (qubit `B') serves as the readout qubit for the square tile ancilla directly to its left and the bottommost one (qubit `C') will be used to perform the readout of the octagon directly below the square tile (not pictured). Because the structure of the $4.8.8.$ code is less amenable to direct mapping the stepping process is a little more complicated. We will again only write down the $Z$-cycle with the $X$-cycle being the same up to initial and final single qubit rotations on the data qubits.\\

\onecolumngrid 

\noindent
\fbox{
\begin{minipage}[t]{\textwidth}
\begingroup
\setlength{\abovedisplayskip}{7pt}
\setlength{\belowdisplayskip}{7pt}
\centering
{\bf The $4.8.8$ color code $Z$-cycle}

\begin{enumerate}[{\bf ~Step 1:}]
\setlength\itemsep{-0.1em}
	\item Initialize in the idle configuration
	\item Apply $ZHS\ct$ to all qubits in $\mc{R}$ (data) and $S\ct$ to qubits in $\mc{B}$ (ancilla)
	\item Go to right square configuration
	\item Go to rightward triangle configuration 
	\item Perform \cnot between qubits A and $1$ and D and $7$ by performing \code{VC[L]} with
	\begin{equation*}
	\code{L} = \{(i,j)\hspace{2mm} \|\hspace{2mm} i=1\mod 2, j = 0 \mod 2, [i+j =3\lor 7 \mod 16]\} 
	\end{equation*}
	\item Perform \cnot between qubits A and $2$ and d and $6$ by performing \code{VC[L]} with
	\begin{equation*}
	\code{L} = \{(i,j)\hspace{2mm} \|\hspace{2mm} i=0\mod 2, j = 0 \mod 2, [i+j =2\lor 6 \mod 16]\} 
	\end{equation*}
	\item Go to left square configuration
	\item Go to left triangle configuration 
	\item Perform \cnot between qubits A and $8$ and D and $9$ by performing \code{VC[L]} with
	\begin{equation*}
	\code{L} = \{(i,j)\hspace{2mm} \|\hspace{2mm} i=1\mod 2, j = 0 \mod 2, [i+j =1\lor 5 \mod 16]\} 
	\end{equation*}
	\item Perform \cnot between qubits A and $7$ and d and $10$ by performing \code{VC[L]} with
	\begin{equation*}
	\code{L} = \{(i,j)\hspace{2mm} \|\hspace{2mm} i=0\mod 2, j = 0 \mod 2, [i+j =0\lor 4 \mod 16]\} 
	\end{equation*}
	\item Go to idle configuration
	\item Go to idle configuration but with all even columns up and all odd columns down by performing \code{VS[L]} with
	\begin{align*}
	\code{L} &= \{(i,j,1)\hspace{2mm} \|\hspace{2mm} i=0\mod 2, j = 0 \mod 2\}\\
				&\hspace{10mm} \cup\{(i,j,-1)\hspace{2mm} \|\hspace{2mm} i=1\mod 2, j = 1 \mod 2\}
	\end{align*}
	\item Go to right square configuration
	\item Go to rightward triangle configuration
	\item Perform \cnot between qubits A and $3$ by performing \code{VC[L]} with
	\begin{equation*}
	\code{L} = \{(i,j)\hspace{2mm} \|\hspace{2mm} i=1\mod 2, j = 0 \mod 2, i+j = 3 \mod 16\} 
	\end{equation*}
	\item Perform \cnot between qubits A and $4$ by performing \code{VC[L]} with
	\begin{equation*}
	\code{L} = \{(i,j)\hspace{2mm} \|\hspace{2mm} i=0\mod 2, j = 0 \mod 2, i+j = 2 \mod 16\} 
	\end{equation*}
	\item Go to idle configuration
	\item Go to left square configuration
	\item Go to leftward triangle configuration
	\item Perform \cnot between qubits A and $6$ by performing \code{VC[L]} with
	\begin{equation*}
	\code{L} = \{(i,j)\hspace{2mm} \|\hspace{2mm} i=1\mod 2, j = 0 \mod 2, i+j = 1 \mod 16\} 
	\end{equation*}
	\item Perform \cnot between qubits A and $5$ by performing \code{VC[L]} with
	\begin{equation*}
	\code{L} = \{(i,j)\hspace{2mm} \|\hspace{2mm} i=0\mod 2, j = 0 \mod 2, i+j = 0 \mod 16\} 
	\end{equation*}
	\item Go to idle configuration
	\item Invert {\bf Step 6} by performing \code{VS[L]} with
	\begin{align*}
	\code{L} &= \{(i,j,-1)\hspace{2mm} \|\hspace{2mm} i=0\mod 2, j = 0 \mod 2\}\\
				&\hspace{10mm} \cup\{(i,j,1)\hspace{2mm} \|\hspace{2mm} i=1\mod 2, j = 1 \mod 2\}
	\end{align*}
	\item Repeat {\bf Steps 2-23} but setting $i \text{~to~} i+2$ and $j \text{~to~} j+1$
	\item Apply $ZHS\ct$ to all qubits in $\mc{R}$ (data) and $S\ct$ to qubits in $\mc{B}$ (ancilla)
	\item Go to measurement configuration 
	\item Perform Pauli Spin Blockade measurement process as described in \cref{subsec:parallel measurements} using qubit B (unit cell to the right) as ancilla for qubit A and using qubit C as ancilla for qubit D
	\item Go to idle configuration
\end{enumerate}
\endgroup
\end{minipage}
}
\\\phantom{end}\\

\twocolumngrid

\section{Discussion}\label{sec:discussion}
In this section we evaluate the mapping of the error corrections codes described above and argue numerically that it is possible to attain the error suppression needed for practical universal quantum computing. We will do this exercise for the planar surface code, as it is the most popular and best understood error correction code. The description given in \cref{sec:surface code implementation} assumes that all operations can be implemented perfectly in parallel. In practice though, for the reasons outlined in \cref{sec:An assembly language for crossbar control} many operations that can in principle be done in parallel will be done in a line-by-line fashion.  Note that for surface code in an array like this, the length of a quadratic grid scales linearly with the code distance as $N=2d+1$. This means that the time performing a surface code cycle and thus the number of errors affecting a logical qubit rises linearly with the code distance and hence this mapping of the surface code will not exhibit an error correction threshold. As a consequence the error probability of the encoded qubit (the logical error probability) cannot be made arbitrarily small but rather will exhibit a minimum for some particular code distance after which the logical error probability will start rising with increasing code distance. The code distance which minimizes the error will depend non-trivially on the error probability of the code qubits. This is not a very satisfactory situation from a theoretical point of view, but from the point of view of practical quantum computation we are not so much interested in asymptotic statements but rather if the logical error probability can be made small enough to allow for realistic computation~\cite{fowler2012surface}. As a target logical error probability we choose $P_L = 10^{-20}$ as at this point the computation is essentially error free (for comparison, a modern classical processor has an error probability around $10^{-19}$~\cite{tezzaron2004soft}). We will use this number as a benchmark to assess if and for what error parameters the surface code mapping in the QDP yields a ``practical" logical qubit. In order to assess this we must consider in more detail the sources of error afflicting the surface code operation on the QDP. We will begin by detailing how the surface code is likely to be implemented in practice on the QDP and afterward we will consider how this impacts the error behavior of the logical surface code qubit. We will distinguish two classes of error sources: operation induced errors and decoherence induced errors. 

\subsection{Practical implementation of the surface code}\label{subsec:practical implementation}
Here we present an mapping of the surface code based on the one presented in \cref{sec:surface code implementation} but differing in the amount of time-steps used to perform certain operations. In particular we choose to do all shuttle and two-qubit-gate operations in a line-by-line manner. This is a specific choice which we expect will work well but variations of this protocol are certainly possible. As mentioned above this will mean that the time an error correction cycler takes will scale will the code distance. This means it is important to keep careful track of the time needed to perform a cycle. We will do this while describing line-by-line operation of the surface code cycle in greater detail below.\\

\noindent In practice we will perform the protocol in \cref{sec:surface code implementation} in the following manner. We begin by performing step $1$ and $2$ for all qubits. Then we apply steps $3-7$ but only to the data and ancilla qubits in the columns $0$ and $1$. Note that after performing these steps on only the first two columns we are back in the \emph{idle} configuration. Now we repeat the previous for columns $2$ and $3$ and so forth until we reach the end of the code surface. Having done these operations we are at the end of step $7$ (go to \emph{idle} configuration) and the grid is the \emph{idle} configuration. We now repeat the same process to perform steps $8-12$ of \cref{sec:surface code implementation}. Next we perform step $13$ which can be done globally. Hereafter we perform step $14$  (ancilla correction) in standard line-by-line fashion. Note that even in an ideal implementation step $14$ has to be done line-by-line in the worst case. After this we perform step $15$ (go to \emph{measurement} configuration) in a line-by line manner and similarly for steps $16$ (PSB/readout procedure) and $17$ (go to \emph{idle} configuration).\\

\noindent Note that in this line-by-line implementation there is a slight asymmetry between the $X$- and $Z$-cycles. Due to the boundary conditions of the surface code the $X$-cycle will involve $d+1$ columns pairs whereas the $Z$-cycle will involve $d-1$ column pairs. However since $(d+1) + (d-1) = 2d $ this is mathematically equivalent to saying that the average cycle involves $d$ column pairs. With this understanding we quite in \cref{tab:gate times} how many time-steps every step in \cref{sec:surface code implementation} takes (split up by gates involved in that step) in this particular implementation of the protocol. Note that in this table we do not specify the order in which the operations happen, only to which step they are associated. We also calculate the amount of time-steps (for different gate types) needed for the full surface code error correction cycle.

\begin{table*}
\resizebox{0.8\textwidth}{!}{%
\begin{tabular}{|c|c|c|c|c|c|c|c|c|c|c|c|c|c|c|c|c|c||c|c|}
\hline
Steps				& 1 & 2 & 3& 4& 5 & 6 & 7& 8& 9 & 10 & 11& 12& 13 & 14 & 15& 16&17& $X$-$Z$ cycle average & Full cycle total \\
\hline
\hline
\sw gate 	& \phantom{$0$} & \phantom{$0$} & \phantom{$0$}& \phantom{$0$}& $2d$ & $2d$ & \phantom{$0$} & \phantom{$0$} & \phantom{$0$} & $2d$  & $2d$ & \phantom{$0$} & \phantom{$0$} &\phantom{$0$} & \phantom{$0$}  & \phantom{$0$} & \phantom{$0$} & $8d$ & $16d$\\
\hline
$Z$ rotation		& \phantom{$0$} & \phantom{$0$} & \phantom{$0$}& \phantom{$0$}& $2d$ & $2d$ & \phantom{$0$}& \phantom{$0$}& \phantom{$0$} & $2d$  & $d$ & \phantom{$0$} & \phantom{$0$} &\phantom{$0$} & \phantom{$0$}  & \phantom{$0$} & \phantom{$0$} & $7d$ & $14d$ \\
\hline
Shuttling			& \phantom{$0$} & \phantom{$0$} & $d$& $d$& \phantom{$0$} & \phantom{$0$} & $d$ & $d$ & $d$ & \phantom{$0$}  & \phantom{$0$} & $d$ & \phantom{$0$} &\phantom{$0$} & $5d$  & $2d$ & $3d$ & $16d$ & $32d$ \\
\hline
Global rotation 	& \phantom{$0$} & $1$ & \phantom{$0$}& \phantom{$0$}& \phantom{$0$} & \phantom{$0$} & \phantom{$0$}& \phantom{$0$}& \phantom{$0$} & \phantom{$0$}  & \phantom{$0$} & \phantom{$0$} & $1$ &\phantom{$0$} & $1$  & \phantom{$0$} & \phantom{$0$} & $3$ & $6$ \\
\hline
Measurement 		& \phantom{$0$} & \phantom{$0$} & \phantom{$0$}& \phantom{$0$}& \phantom{$0$} & \phantom{$0$} & \phantom{$0$}& \phantom{$0$}& \phantom{$0$} & \phantom{$0$}  & \phantom{$0$} & \phantom{$0$} & \phantom{$0$} &\phantom{$0$} & \phantom{$0$}  & $d$ & \phantom{$0$} & $d$ & $2d$\\
\hline
\end{tabular}
}
\caption{Time-step count per step in terms of different types of possible gates for the line-by-line implementation of the surface code cycle described in \cref{sec:surface code implementation}. The number of time-steps is quoted in terms of the code distance $d$. this table does not specify the exact order in which the operations happen, see \cref{subsec:practical implementation} for an explanation of the time flow. Note that the table shows the average of the time-step counts for the $X$- and $Z$-cycles. The actual time count for the individual $X$- and $Z$-cycles is slightly different due to the boundary conditions of the surface code. The exact count for the $Z$-cycle can be obtained by replacing $d$ by $d-1$ in every entry (except for the last column) whereas the exact count for the $X$-cycle is obtained by replacing $d$ with $d+1$ in every column bar the last one. Since $(d+1) + (d-1) = 2d = d+d$ this makes no difference for the full cycle count. Table cells that are left empty signify zero entries.}\label{tab:gate times}
\end{table*}

\subsection{Decoherence induced errors}
Decoherence induced errors are introduced into the computation by uncontrolled physical processes in the underlying system. The effect of these processes is called decoherence. Decoherence happens even if a qubit is not being operated upon and the amount of decoherence happening during a computation scales with the time that computation takes. Therefore, to account for decoherence induced errors during the error correction cycle we need to compute how long an error correction cycle takes. Generally any operation on the QDP takes a certain amount of time denoted by $\tau$. We distinguish again five different operations: (1) two-qubit \sw gates, (2) qubit shuttle operations, (3) single qubit $Z$ gates by waiting, (4) global single qubit operations, (5) qubit measurements. The time they take we will denote by $\tau_{sw},\tau_{sh},\tau_{z}, \tau_{gl}$ and $\tau_{m}$ respectively. In \cref{tab:gate times} we performed a count of the total time taken by the surface code error correction cycle using the mapping described in \cref{sec:surface code implementation,subsec:practical implementation}.  The table below summarizes the total number of time-steps for every gate type for a full surface code error correction cycle.
\begin{figure}[ht]\label{fig:operation times}
\centering
\begin{tabular}{|c|c|c|c|}
\hline
Symbol & Operation & time-steps per cycle\\
\hline
$\tau_{sw}$ & \sw gate &$16d$\\
$\tau_{sh}$ & Shuttling & $32d$\\
$\tau_{z}$ & $Z$ rotation by waiting & $14d$\\
$\tau_{gl}$ & Global qubit rotation &$6$\\
$\tau_{m}$ & Measurement & $2d$\\
\hline
\end{tabular}
\end{figure}
We can now say the total time $\tau_{\mathrm{total}}(d)$ as a function of the code distance $d$ is given by
\begin{equation}
\begin{aligned}
\tau_{\mathrm{total}}(d) &= 16d\tau_{sw} + 32d\tau_{sh} \\&\hspace{10mm}+ 14d\tau_{z} +6\tau_{gl} +2d\tau_{m}.
\end{aligned}
\end{equation}
This total time can be connected to an error probability by invoking the mean decoherence time of the qubits in the system, the so called $T_2$ time~\cite{nielsen2002quantum,tomita2014low} (We ignore the influence of $T_1$ in this calculation as it is typically much larger than $T_2$ in silicon spin qubits~\cite{veldhorst2017crossbar,tyryshkin2011electron}). We can find the decoherence induced error probability $P_{dec}$~\cite[Page 384]{nielsen2002quantum} as
\begin{equation}\label{eq:time to decoherence}
P_{dec}(d) = \frac{\tau_{\mathrm{total}}(d)}{2T_2}.
\end{equation}
Next we investigate operation induced errors. These will typically be larger than decoherence induced errors but will not scale with the distance of the code.

\subsection{Operation induced errors}
Operation induced errors are caused by imperfect application of quantum operations to the qubit states. There are five operations performed on qubits in the surface code cycle. These are: (1) two-qubit \sw gates, (2) qubit shuttle operations, (3) single qubit $Z$ gates by waiting, (4) global single qubit operations, (5) qubit measurements. We will denote the probability of an error afflicting these operations by $P_{sw},P_{sh},P_{z}, P_{gl}$ and $P_{m}$ respectively. In \cref{tab:gate totals} we list the total number of gates of a given type a data qubit and an ancilla qubit participate in over the course of a surface code cycle. In \cref{appsec:operation counts} we give a more detailed per-step overview of the operations performed on data qubits and ancilla qubits. For clarity we have chosen qubit $1$ in \cref{fig:unit cells} (right) as a representative of the data qubits and qubit $A$ in \cref{fig:unit cells} (right) as a representative of the ancilla qubits. Other qubits in the code might have a different ordering of operations but their gate counts will be the same, except for the qubits located at the boundary of the code which will have a strictly lower gate count (we can thus upper bound their operation induced errors by those of the representative qubits). For each gate we also calculate the average number of this gate data and ancilla qubits participate in. This average number will serve as our measure of operation induced error.

\begin{table*}
\resizebox{0.8\textwidth}{!}{%
\begin{tabular}{|c|c|c|c||c|c|c||c|}
\hline
 & \multicolumn{3}{c||}{Data qubit} & \multicolumn{3}{c||}{$Z$ ancilla qubit}& \multirow{2}{*}{Average data/ancilla}\\
 \cline{2-7}
	& $Z$-cycle & $X$-cycle & Total&  $Z$-cycle & $X$-cycle & Total& \\
\hline
\sw gate 	& $4$ & $4$ & $8$& $8$& $0$ & $8$ & $8$  \\
\hline
$Z$ rotation		& $0$ & $0$ & $0$& $7$& $0$ & $7$ & $3.5$ \\
\hline
Shuttling			& $2$ & $4$ & $6$& $10$& $4$ & $14$ & $10$ \\
\hline
Global rotation 	& $2$ & $2$ & $4$& $2$& $3$ & $5$ & $4.5$  \\
\hline
Measurement 		& $0$ & $0$ &$0$& $1$& $1$ & $2$ & $1$   \\
\hline
\end{tabular}
}
\caption{This table lists the total number of gates of a given type a data qubit and an ancilla qubit participate in over the course of a surface code cycle. In \cref{appsec:operation counts} we give a more detailed per-step overview of the operations performed on data qubits and ancilla qubits. For clarity we have chosen qubit $1$ in \cref{fig:unit cells} (right) as a representative of the data qubits and qubit $A$ in \cref{fig:unit cells} (right) as a representative of the ancilla qubits. Other qubits in the code might have a different ordering of operations but there gate counts will be the same, except for the qubits located at the boundary of the code which will have a strictly lower gate count (we can thus upper bound their operation induced errors by those of the representative qubits). }\label{tab:gate totals}
\end{table*}

\subsection{Surface code logical error probability}\label{subsec: surface code logical error probability}
By tallying up the contributions from operational and decoherence induced errors we can construct a measure for the total error probability per QEC cycle experienced by all physical qubits that make up the code. Note that this a rather crude model that disregards possible influences from inter-qubit correlated errors and time-like correlated errors. Nevertheless it serves as a useful first approximation to the performance of the surface code on the QDP. We define the average per qubit per cycle error probability $P_{\mathrm{tot}}$ as
\begin{equation}\label{eq:total error}
\begin{aligned}
P_{\mathrm{tot}}(d) &= 8P_{sw} + 3.5P_{sh} + 10P_{z} \\&\hspace{10mm}+ 4.5P_{gl} + P_{m} + P_{dec}(d).
\end{aligned}
\end{equation}
Note that this quantity depends linearly on the code distance $d$. We can plug this total per cycle error probability $P_{tot}$ into an empirical equation for the logical error probability $P_L$ derived in~\cite{fowler2012surface}. 
\begin{equation}\label{eq:empirical model}
P_L = 0.03\left(\frac{P_{tot}(d)}{8P_{th}}\right)^{\frac{d+1}{2}}
\end{equation}
where $P_{th}$ is the per-step fault-tolerance threshold of the surface code, which we take to be $P_{th} = 0.0057$ following the result in~\cite{fowler2012surface}. The factor of $8$ is inserted to account for the fact that the empirical relation derived in~\cite{fowler2012surface} is between the physical \emph{per-step} error rate and the logical \emph{per cycle} error rate and the protocol analyzed in~\cite{fowler2012surface} requires $8$ time-steps per surface code error correction cycle. This is an approximation but it will serve our purposes of getting a basic initial estimate of the logical error rate. The next step is to start plugging in experimental numbers into equation \cref{eq:total error}. In the table below we quote error probabilities and operation times for all relevant parameters. These numbers are projections from~\cite{veldhorst2017crossbar} and references therein.
\begin{figure}[ht]\label{fig: experimental numbers}
\centering
\begin{tabular}{|c|c|c|c|}
\hline
Operation & Error probability& Time\\
\hline
two-qubit \sw gate & $P_{sw} = 10^{-3}$ & $\tau_{sw} = 20\ns$\\
qubit shuttle & $P_{sh} = 10^{-3}$ & $\tau_{sh} = 10\ns$\\
$Z$ rotation by waiting& $P_{z} = 10^{-3}$ & $\tau_{z} = 100\ns$\\
global qubit rotation& $P_{gl} = 10^{-3}$ & $\tau_{gl} = 1000\ns$\\
measurement& $P_{m} = 10^{-3}$ & $\tau_{m} = 100\ns$\\
\hline
\end{tabular}
\end{figure}
To convert the operation times into decoherence induced error we use the estimated $T_2$ time of quantum dot spin qubits in $^{28}\mathrm{Si}$ quoted as $T_2 = 10^{9}\ns$~\cite{veldhorst2017crossbar,tyryshkin2011electron} and \cref{eq:time to decoherence}. 
Plugging these numbers into \cref{eq:total error} we get the following linear function of the code distance
\begin{equation}\label{eq:ptot}
P_{tot} = 2.7\times 10^{-2}+ 2.8d\times 10^{-5}
\end{equation}
which we can plug into the empirical model \cref{eq:empirical model}. In \cref{fig:logical error plot} we plot the logical error probability $P_{L}$ versus code distance. Note that for the experimental numbers provided the practical quantum computing benchmarking $\log(P_{L}) = -20$ is reached for a code distance of $d=37$. The maximal code distance for the experimental parameters is $d=155$ for which the log-logical error probability reaches $\log(P_{L}) = -41$, after which it starts increasing again. For completeness we have also plotted what would happen if we had the power to operate the QDP (with quoted device parameters) completely in parallel. We estimate the physical per cycle error rate of this situation by setting $d=1$ in \cref{eq:ptot}. Note that the difference between parallel and crossbar style operation is not that big, the parallel version reaches $P_L = 10^{-20}$ for $d= 31$. This rough model provides some quantitative justification for the implementation of planar error correction codes in the QDP even in the absence of the ability to arbitrarily suppress logical error. Note also that, due to the long coherence times~\cite{veldhorst2017crossbar,tyryshkin2011electron} of the QDP spin qubits, the dominant terms in the expression for the total error probability $P_{tot}$ are those associated with operation induced errors. This provides justification for the line-by-line application of two-qubit gates discussed in \cref{subsec:parallel two-qubit gates}, which takes a longer time to perform but improves gate quality. It also means that long coherence times and/or fast operation times are likely critical to the success of a crossbar based scheme. This concludes our discussion of the QDP mapping of the surface code. A similar exercise can be done for the $6.6.6.$ and $4.8.8.$ color codes but due to their lower thresholds~\cite{landahl2011fault}, the results will likely be less positive for current experimental parameters.
\begin{figure}[ht]
\centering
\includegraphics[scale=0.57]{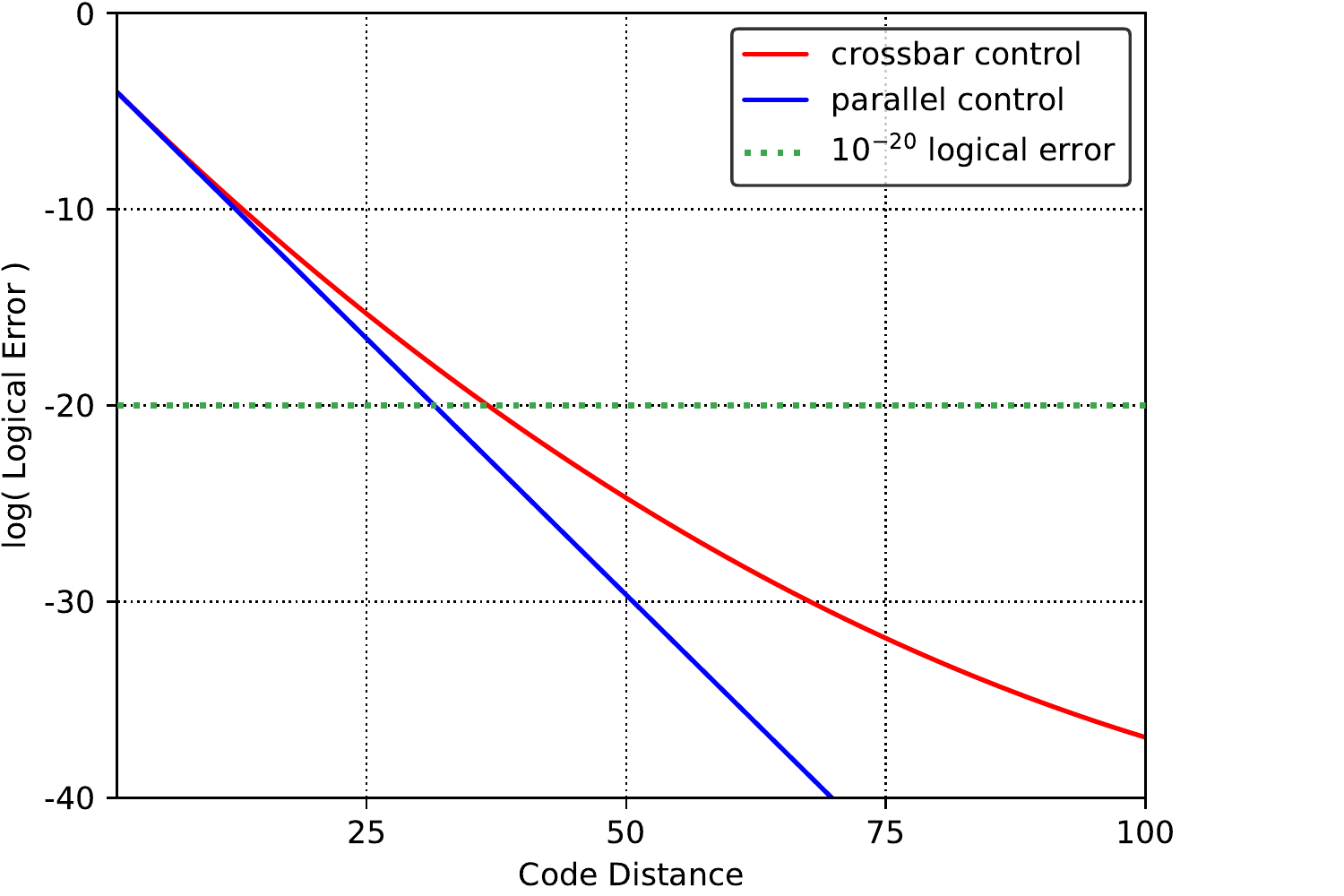}
\caption{Plot of logical error probability versus code distance for the empirical model given in \cref{eq:empirical model} with experimental parameters given in \cref{fig: experimental numbers}. Note that the logical error probability for crossbar operation goes below $P_L = 10^{-20}$ for $d=37$. This is only slightly slower that parallel operation, which reaches $P_L = 10^{-20}$ for $d=31$. Due to the scaling of crossbar operation with the code distance the logical error probability bottoms out at some point. This however does not happen until $d=155$ (not shown) for a logical error rate of $P_L = 10^{-41}$, which is not practically relevant. This rough model gives good indication it is possible to create very low logical error surface code logical qubits in the QDP.}\label{fig:logical error plot}
\end{figure}

\section{Conclusion}\label{sec:conclusion}
We analyzed the architecture presented in \cite{veldhorst2017crossbar}, focusing on its crossbar control system. Building on this analysis we presented procedures for mapping the planar surface code and the $6.6.6.$ and $4.8.8.$ color codes. Because the line-by-line operation of the crossbar architecture means the noise in a single error correction cycle scales with the distance it is not possible to arbitrarily suppress the logical error rate by increasing the code distance. Instead there will be some ``optimal" code distance for which the logical error rate is the lowest. Using numbers for~\cite{veldhorst2017crossbar} and an empirical model taken from~\cite{fowler2012surface} we analyzed the logical error behavior of the surface code mapping and found that, for current experimental numbers, it is at least in principle possible to achieve logical error probabilities below $P_{log} = 10^{-20}$, making practical quantum computation possible. However, we stress that this is a rather crude estimate and a more detailed answer would have to take into account the details of the dominant error processes in quantum dot qubits. It must also take into account that while it is possible to achieve certain low noise gates  and good coherence times in quantum dots qubits in isolation this does not necessarily mean they will be practically achievable in the current QDP design.\\

\noindent In future work we would like to use the currently developed machinery to map more exotic quantum error correction codes. Due to the possibility of qubit shuttling, codes with long distance stabilizers could in principle be implemented. Codes such as the 3D gauge color codes might be prime candidates for this kind of treatment. However, barring some special cases, parallel shuttling is currently being performed in a line-by-line manner. A general classical algorithm for generating optimal (in time) shuttling-steps from an initial to a final \boardstate would vastly simplify the task of mapping more exotic codes and also general quantum circuits. Such an algorithm would probably be useful for any crossbar quantum architecture. In this work we constructed a non-optimal but classically efficient algorithm but finding an algorithm that generates optimal shuttling sequences and analyzing its resource use is still an open problem.

\subsection*{Acknowledgments}
We would like to thank J\'er\'emy Ribeiro, David Elkouss, Kobe Vrancken, Roy Li and Lieven Vandersypen for helpful suggestions, discussions and feedback. We would like to thank C.~W.~J.~Beenakker for his support. JH and SW 
are supported by STW Netherlands, NWO VIDI and an ERC Starting Grant. MS is supported by the Netherlands Organization for
Scientific Research (NWO/OCW) and an ERC Synergy Grant. MV acknowledges support by the Netherlands Organization of Scientific Research (NWO) VIDI program.

\bibliography{crossbarbibliography.bib}

\newpage
\onecolumngrid

\appendix

\section{Shuttling algorithm}\label{appsec:Shuttling algorithm}
In this appendix we go a little deeper into the shuttling algorithm presented in \cref{subsec:parallel shuttle operations}. This algorithm takes as input a collection of desired shuttling operations in the form of a flow matrix $F$ and outputs a sequence of parallel \code{HS} operations that, when applied sequentially, achieve this desired collection of operations. The algorithm, described in~\Cref{alg:shuttling algorithm} relies centrally on a notion of independence on the columns of this flow matrix $F$. This notion of independence is actualized by calling an `independence subroutine' through the functions \code{CheckIndependence} and \code{DependenceSet}. Here we describe various independence subroutines and analyze their time complexity. Note that it is probably possible to optimize these subroutines and the time complexity estimates, hence the results in \cref{thm:simple subroutine complexity,thm:k commuting subroutine complexity,thm:greedy subroutine complexity} can best be seen as upper bounds on the worst case time complexity. Recall from \cref{subsec:parallel shuttle operations} that the flow matrix $F$ has entries in the set $e,r,l, re, le, *$ where $e$ signifies doing nothing, $r$ signifies a rightward shuttling operation, $l$ signifies a leftward shuttling operation and $re, le$ and $*$ are `wildcard' symbols that indicate the operation at the point could be either $r$ or $e$, or $l$ or $e$, or $r$ or $l$ or $e$. We begin by analyzing the mathematical structure of the shuttle operations $e,r,l$.
\subsection{The left-right monoid}
An idempotent monoid $\{M,\circ\}$ is a set $M$ with a binary operation $\circ:M\times M \rightarrow M$ such that the following axioms hold
\begin{align}
&\forall a,b,c \in M:\;\;\;\;\;\;\;\;\;\;\;\;\;\:(a\circ b)\circ c = a\circ(b\circ c), \tag{Associativity}\\
&\exists e \in M, \forall a \in M:\;\;\;\;\;\;\; e\circ a = a\circ e = a, \tag{Identity element}\\
&\forall a \in M:\;\;\;\;\;\;\;\;\;\;\;\;\;\;\;\;\;\;\;\; a\circ a = a. \tag{Idempotence}
\end{align}
Note that monoids are strict generalizations of groups (as the elements lack an inverse). We will argue that the set of shuttle operations $r,l,e$ together with the binary operation `composition of shuttle operations' is an idempotent monoid. Imagine for simplicity a \boardstate with only one row and two columns. The shuttle operations that can be applied to this system are `shuttle to the right' (r), shuttle to the left (l) and `do nothing' (e). These shuttle operations can also be applied sequentially. We will denote the sequential application of operations $a_1,a_2$ as $a_1\circ a_2$. This is read from the right, so we apply first $a_2$ and then $a_1$. For example shuttling to the right at a location and subsequently doing nothing at that location will be written as $e\circ r$. Note that this is equivalent to just shuttling to the right at that location so we have $e\circ r = r$. Other examples are more interesting.For instance shuttling to the right at a fixed location followed by shuttling to the left at that same location is equivalent to just shuttling to the left at that location. However shuttling to the left first and then shuttling to the right is equivalent to shuttling to the right. this means we have $r\circ l = r \neq l = l \circ r$ which means the operation $\circ$ is not commutative. Not also that shuttling to the right at a fixed location and then shuttling to the right again at that location is equivalent to shuttling to the right a single time at that location. Hence we have$r\circ r=r$. In general we have the following rules for the composition of the operations $e,r,l$
\begin{gather}\label{eq:monoid rules}
r\circ r = r,\hspace{10mm} l\circ l = l,\\
r\circ l = r,\hspace{10mm}l\circ r =r, \\
r\circ e = r,\hspace{10mm} l\circ e =l,\\
e\circ r = r,\hspace{10mm} e \circ l =e.
\end{gather}
Note that these rules imply that the composition $\circ$ is also associative. This makes the set $S = \{e,r,l\}$ with the composition $\circ$ an idempotent monoid with $e$ the identity element. Note also that neither $r$ nor $l$ is has an inverse. We call the idempotent monoid $\{S,\circ\}$ the left-right monoid. The non-commutativity of $\circ$ and the fact that neither $r$ nor $l$ have inverses makes finding a sequence of parallel shuttling operations that apply the shuttling operations encoded in the flow matrix $F$ difficult in general as the order in which the operations are applied matters and no operation can be truly inverted. 

\subsection{Comparing up to wildcards}
Recall that the entries of the flow matrix $F$ take value in the set $\{r,l,e,re,le,*\}$. Of these elements only the first three ($\{r,l,e\}$) correspond to real actions. The other elements $\{re,le,*\}$ are called wildcard elements. If an entry in the flow matrix is wildcard valued it means we have some freedom which of the operations $\{r,l,e\}$ we apply there. When constructing the independence subroutines we need a way to compare elements of the flow matrix that takes this freedom into account. To this end we introduce the relation `equality up to wildcards', signified by the symbol $=_w$. This relation formalizes the the intuitive notion that~e.g. the elements $re$ and $r$ but also the elements $re$ and $e$ should be equal since if an entry of the flow matrix $F$ takes the value $re$ we can perform either the operation $r$ or $e$ without performing an illegal move. Below we list all elements that are equal up to wildcards (up to symmetry and reflexivity)
\begin{align}
&* =_w r\hspace{10mm} re=_w r \hspace{10mm} le =_w l\\
&* =_w l\hspace{10mm} re=_w e \hspace{10mm} le =_w e\\
&* =_w e\\
&* =_w re\\
&* =_w le.
\end{align}
\subsection{Comparing columns of the flow matrix}
Checking an equality up to wildcards computationally takes at most five list comparisons.
We would also like to be able to compare columns of the flow matrix $F$. This is so because repeated patterns in the flow matrix columns can be performed in a single parallel shuttle operation. To see why this is the case consider the following example flow matrix
\begin{equation}\label{eq: equal columns}
F_{ex} = \begin{pmatrix} *&*&*&* \\ *& re & r&* \\ * & re & r&* \\ r&*&*&*\\ r & * & * & l \end{pmatrix} \;\;\;\;\;\;\;\;\;\;\;\raisebox{-0.125\textwidth}{\includegraphics[scale=1]{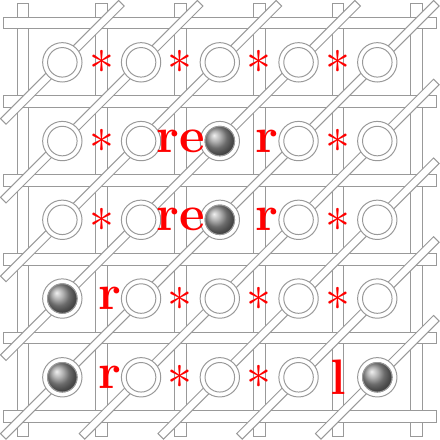}}
\end{equation}
For illustration we also included a \boardstate configuration that could have led to this particular flow matrix. Here it is clear how all elements of the flow matrix correspond to a crossing point on the grid. Notice that columns $0$ and $2$ of $F_{ex}$ are very similar. In fact they are the same up to a constant shift. Note also that we can apply all $r$ operations encoded in the flow matrix by performing the operation \code{HS[\{(0,0,1), (1,0,1,),(2,2,1),(2,3,1)\}]}. This operation can be done in a single time-step by the parallel control OPCODE \code{V[0]\&V[1]\&D[0][0]\&D[1][1]\&D[3][2]}. This can be done because the two columns are equal up to a constant vertical shift. Because the diagonal lines run in a $45$ degree angle over the QDP this shift is exactly equal to the difference between the column indices of the two columns. In our example the shift is $2$ since we are comparing the columns $0$ and $2$ but in general when checking if column $i$ and $j$ can be performed in parallel we must check if they are equal (up to wildcards, see above) up to a constant vertical shift of size $i-j$. The facilitate this process we define the padding function $p$:
\begin{align}
p:\{r,l,e,re,le,*\}^{\times(N-1)}\times[1:N] \longrightarrow& \{r,l,e,re,le,*\}^{\times(2N-2)}:\\
 &p(v,i) \longrightarrow (*,\ldots,\underset{i-1}{*},v,\underset{N+i}{*}, \ldots, *).
\end{align}
Now because, up to wildcards, the $*$ symbol is equal to all symbols other we can check if two columns $v_i,v_j$ of $F$ are equal up to wildcards and a constant vertical shift by checking if the padded columns $p(v_i,i),p(v_j,j)$ are equal up to wildcards. This means checking $p(v_i,i)_t=_wp(v_j,j)_t$ for all $t\in [1:2N]$ where $p(v_i,i)_t$ is the $t$'th component of the padded column $p(v_i,i)$.

\subsection{Composition of columns in the flow matrix}
Sometimes it happens that columns in the flow matrix can be written as the composition of other columns. This is the case in the following example
\begin{equation}\label{eq:comm columns}
F_{ex} = \begin{pmatrix} * & * & l \\ * & * & r \\ l & r & * \end{pmatrix}.
\end{equation}
Note that we have $p(v_0,0)\circ p(v_1,1) =_w p(v_2,2)$ where $v_i$ is the $i$'th column of $F_{ex}$ and the composition and equality up to wildcards are taken element-wise. This means that in a sense $v_2$ is `dependent' on $v_0$ and $v_1$. This means we can apply the operations encoded by the flow matrix $F$ in two sequential parallel shuttling steps by duplicating the operations encoded in columns $0$ and $1$ to also cover column $3$ (up to the correct constant vertical shift). The shuttle operations that need to be performed to apply the shuttle operations encoded in the flow matrix are
\begin{gather}
\code{HS[(0,0,-1), (2,2,-1)]}\\
\code{HS[(0,1,1), (1,2,1)]}.
\end{gather}
Note that in this particular case we have $p(v_0,0)\circ p(v_1,1) =_w p(v_1,1)\circ p(v_0,0)$, i.e.~$p(v_0,0)$ and $p(v_1,1)$ commute with respect to the composition `$\circ$', which means we can apply the operations above in any order. This does not need to be the case. Consider for instance in the following example:
\begin{equation}\label{eq:noncomm columns}
F_{ex} = \begin{pmatrix} * & * & l \\ * & r & r \\ l & r & * \end{pmatrix}.
\end{equation}

Note that we now have $p(v_0,0)\circ p(v_1,1) =_w p(v_2,2) $ but not $p(v_1,1)\circ p(v_0,0) =_w p(v_2,2)$! This means that in this case we can still apply all shuttling operation encoded in the flow matrix $F_{ex}$ by extending the operations that apply the shuttle operations encoded in $v_0$ and $v_1$ to also include the operations encoded in $v_2$ but now the order in which we perform the operations matters. We have to first apply the operation $\code{HS[(0,1,1),(1,1,1),(1,2,1),(2,2,1)]}$ and then $\code{HS[(0,0,-1), (2,2,-1)]}$ in order to apply the operations encoded in the flow matrix $F_{ex}$. The fact that the ordering of the operations matter is a difficulty we have not fully overcome. Therefore we will restrict ourselves only to compositions of columns that explicitly commute. Note that this means for two columns $v_i,v_j$ and $k\in [1:2N-2]$ that if $p(v_i,i)_k =_w r$ we must have $p(v_j,j)_k =_w e$ or $p(v_j,j)_k =_w r$ or vice versa. A similar rule holds if $p(v_j,j)_k =_w l$. Using the above analysis we now present two simple subroutine algorithms that decide whether a column is dependent, i.e.~can be written as a restricted composition of other columns in $F$ (up to wildcards and vertical shifting). 

\subsection{Simple subroutine}
This subroutine simply detects whether the columns of a flow matrix $F$ have duplicates up to wildcard symbols. This means the restriction of the previous section is basically the strictest possible, columns can only be written as compositions of something exactly equal. See \cref{eq: equal columns} for an example of a flow matrix with this property. We have two distinct subroutine functions. The first $\mathbf{CheckIndependence}(S,v_i)$ checks whether a column $v_i$ of the flow matrix $F$ is independent of the columns in the set $S$ while the second, $\mathbf{DependenceSet}(S,v_i)$ returns a set of columns $A_i$ on which $v_i$ depends. For the $\mathrm{[simple]}$ subroutines this will be a set with a single entry, namely an exact copy (up to wildcards) of $v_i$ in $S$. The two subroutine functions are given in \cref{alg:simple checkindependence,alg: simple dependenceset} and their time complexity is analyzed in \cref{thm:simple subroutine complexity}.
\begin{algorithm}[H]
\caption{CheckIndependence [simple]}
\label{alg:simple checkindependence}
\begin{algorithmic}[1]

\Require{A column $v_i$, a set of columns $S$, the column index $i$ of the column $v_i$ }

\Ensure{Boolean $a$} \\

\State // We will consistently write columns of the flow matrix $F$ as $v_i$ where $i$ indicates the column index of $v_i$ in $F$. \\

\For{all columns $c_j\in S$}

	\If{$p(v_i,i) =_w p(c_j,j)$}\\

		\State // Note that all the elements of $C$ commute. This means that if $p(v_i,i) =_w p(c_j,j)$ for a particular $c_j$ it \State // must also hold that $p(v_i,i) =_w p(c_k,k)$ for all $c_k \in C$. \\

		\State \textbf{Set} $a$ to TRUE

	\Else

		\State \textbf{Set} $a$ to FALSE

	\EndIf

\EndFor

\State\Return{$a$}

\end{algorithmic}
\end{algorithm}

\begin{algorithm}[H]
\caption{DependenceSet [simple]}
\label{alg: simple dependenceset}
\begin{algorithmic}[1]

\Require{A column $v_i$, a set of columns $S$, the column index $i$ of the column $v_i$ }

\Ensure{Set of columns $A_i$ that $v_i$ depends on}\\

\State // We will consistently write columns of the flow matrix $F$ as $v_i$ where $i$ indicates the column index of $v_i$ in $F$.

\State // We also tag the output set $A_i$ with the subscript $i$ to indicate it is connected to the $i$'th column $v_i$ of the \State // flow matrix $F$. \\

\State \textbf{Set} $A_i$ to an empty set

\For{ all columns $c_j\in S$}

	\If{$p(v_i,i) =_w p(c_j,j)$}

		\State Add $c_j$ to $A_i$

	\EndIf
	
\EndFor	
	
\State\Return{$A_i$}

\end{algorithmic}

\end{algorithm}

\begin{theorem}\label{thm:simple subroutine complexity}
The subroutines $\mathbf{CheckIndependence(~)}\mathrm{~[simple]}$ and $\mathbf{IndependenceSet(~)}\mathrm{~[simple]}$ have time complexity
\begin{align}
O\big(\mathbf{CheckIndependence(~)}\mathrm{~[simple]}\big) &= O(NM)\\
O\big(\mathbf{IndependenceSet(~)}\mathrm{~[simple]}\big) &= O(NM)
\end{align}
where $M = |S|$ the size of the independent set $S$ and $N$ is the length of the input column $v_i$ or equivalently the number of rows in the flow matrix $F$. Note that for our purposes we have $M\leq N$.
\end{theorem}
\begin{proof}
The complexity of the subroutine $\mathbf{CheckIndependence(~)}\mathrm{~[simple]}$ can be seen by straightforward counting. There is one $\mathbf{For}$-loop (line 1) of length $M$ and the $\mathbf{If}$-clause on line 2 takes $O(N)$ time to evaluate since $c_j$ is a column of length $O(N)$. This gives a total worst case complexity of $O(NM)$. Exactly the same argument holds for the subroutine $\mathbf{IndependenceSet(~)}\mathrm{~[simple]}$.
\end{proof}
\newpage

\subsection{k-commuting subroutines}
Next we present a class of subroutines collectively called the `commuting' subroutines. These try to capture the intuition behind the example flow matrix in \cref{eq:comm columns}, namely that some columns of the flow matrix can be written as a composition of a subset of pairwise commuting columns of the flow matrix. By a pairwise commuting subset $C$ we mean concretely that for all columns $v_i,v_j$ of $F$ that are in $C$ we have $p(v_i,i) \circ p(v_j,j) =_w p(v_j,j) \circ p(v_i,i)$. We restrict explicitly to pairwise commuting columns since this avoids the difficulty of time-ordering the resulting shuttle operations (as seen in \cref{eq:noncomm columns}). This subroutine relies on an initial construction of all maximal mutually commuting subsets of a set of columns $S$. Listing all these sets is in general hard. To see this we can construct the $M\times M$ matrix $A$ that has entries $A_{pq}=1$ whenever the $p$'th and $q$'th column in the set $S$ commute and $A_{pq}=0$ otherwise. If we now think of $A$ as the adjacency matrix of a graph $G$ with $M$ vertices it is not hard to see that finding all maximal mutually commuting subsets of $S$ is equivalent to finding listing all maximal cliques~\cite{west2001introduction} in the graph $G$. The best known algorithm for listing all maximal cliques in an arbitrary graph is called the Bron-Kerbosch algorithm~\cite{bron1973algorithm} and has a worst-case complexity of $O(3^{M/3})$. There is no a priori way to restrict the number of possible maximal cliques generated by the commutation rules of the monoid-valued columns so currently any subroutine that searches over all possible maximal mutually commuting subsets of $S$ will take at least $O(3^{M/3})$ time. We can however get out of this bind by restricting the size of the sets of mutually commuting columns to be less than a fixed parameter $k$. This will make our algorithm less effective but the problem of finding these reduces to finding all cliques of size less than $k$ in the graph $G$ which is a so called fixed-parameter-tractable problem. This problem has worst-case time complexity upper bounded by $O(M^k k^3)$~\cite{downey1995fixed}. Using these sets we can construct a family of subroutines indexed by the parameter $k$. The two subroutine functions are given in \cref{alg:k commuting independenceset,alg:k commuting checkindependence} and their complexity is analyzed in \cref{thm:k commuting subroutine complexity}.

\begin{algorithm}[H]
\caption{CheckIndependence [$k$-commuting]}
\label{alg:k commuting checkindependence}
\begin{algorithmic}[1]

\Require{A column $v_i$, a set of columns $S$, the column index $i$ of the column $v_i$ and a fixed integer $k$}

\Ensure{Boolean $a$} \\

\State // We will consistently write columns of the flow matrix $F$ as $v_i$ where $i$ indicates the column index of $v_i$ in $F$.\\

\State Construct All mutually commuting subsets $C$ of $S$ with $|C|\leq k$

\State Set $a$ to FALSE 

\For{all commuting subsets $C$}

	\For{$t \in [0:\mathrm{length}(p(v_i,i))-1]$}

		\For{all columns $c_j \in C $}

			\If{$p(c_j,j)_t =_w p(v_i,i)_t$ }\\

				\State // Note that all the elements of $C$ commute. This means that if $p(v_i,i) =_w p(c_j,j)$ for a particular  \State // $c_j$ it must also hold that $p(v_i,i) =_w p(c_k,k)$ for all $c_k \in C$.\\

				\If{$t=\mathrm{length}(p(v_i,i))-1$}

					\State Set $a$ to TRUE

				\EndIf

			\Else

				\State Go to next commuting subset $C$ in the loop at line 3

			\EndIf

		\EndFor

	\EndFor

\EndFor

\State \Return{$a$}

\end{algorithmic}	
\end{algorithm}

\begin{algorithm}[H]
\caption{DependenceSet [$k$ - commuting]}
\label{alg:k commuting independenceset}
\begin{algorithmic}[1]

\Require{A column $v_i$, a set of columns $S$, the column index $i$ of the column $v_i$ }

\Ensure{Set of columns $A_i$ that $v_i$ depends on}\\

\State // We will consistently write columns of the flow matrix $F$ as $v_i$ where $i$ indicates the column index of $v_i$ in $F$. \\

\State Construct All maximal mutually commuting subsets $C$ of $S$ with $|C|\leq k$ 

\For{all commuting subsets $C$}

	\For{$t \in [0:\mathrm{length}(p(v_i,i))-1]$}

		\For{all columns $c_j \in C $}

			\If{$p(c_j,j)_t =_w p(v_i,i)_t$ }\\

				\State // Note that all the elements of $C$ commute. This means that if $p(v_i,i) =_w p(c_j,j)$ for a particular \State // $c_j$ it must also hold that $p(v_i,i) =_w p(c_k,k)$ for all $c_k \in C$. \\

				\State \textbf{Add} $c_j$ to $A_i$

				\If{$t=\mathrm{length}(p(v_i,i))-1$}

					\State\Return{$A_i$} 

				\Else

					\State \textbf{Set} $A_i$ to empty set

					\State Go to next commuting subset $C$ in the loop at line 2

				\EndIf	

			\EndIf

		\EndFor

	\EndFor

\EndFor
\State\Return{$A_i$}
\end{algorithmic}
\end{algorithm}

\begin{theorem}\label{thm:k commuting subroutine complexity}
The subroutines $\mathbf{CheckIndependence(~)}\mathrm{~[k-commuting]}$ and\\\noindent $\mathbf{IndependenceSet(~)}\mathrm{~[k-commuting]}$ have time complexity
\begin{align}
O\big(\mathbf{CheckIndependence(~)}\mathrm{~[k-commuting]}\big) &= O(N M^{k+1} k^4 )\\
O\big(\mathbf{IndependenceSet(~)}\mathrm{~[k-commuting]}\big) &= O(N M^{k+1} k^4 )
\end{align}
where $M = |S|$ the size of the independent set $S$, $N$ is the length of the input column $v_i$ or equivalently the number of rows in the flow matrix $F$and $k$ is a fixed parameter indicating the maximal size of the mutually commuting subsets $C$. Note that for our purposes we have $M\leq N$.
\end{theorem}
\begin{proof}
We can again find the complexity of the subroutine $\mathbf{CheckIndependence(~)}\mathrm{~[k-commuting]}$ by a counting argument. We have already noted that the construction in line $1$ takes $O(M^k k^3)$ time in the worst case. Apart from that we have a $\mathbf{For}$-loop on line 3 that takes worst-case time $O(M^k k^3)$ to loop over and given that $|C|\leq k$ we can see that the $\mathbf{For}$-loops on line 4 and 5 take respectively $O(N)$ and $O(k)$ time to complete. Finally the $\mathbf{If}$-clause on line 6 takes $O(N)$ time to complete. Tallying this up we get a total worst-case time complexity of $O(N M^{k+1} k^4 )$. We can make the same argument for the worst-case time complexity of $\mathbf{IndependenceSet(~)}\mathrm{~[k-commuting]}$.
\end{proof}
\newpage

\subsection{Greedy commuting subroutine}
Building on the last subsection we introduce one last pair of subroutines dubbed greedy subroutines. As seen in the last section it seems hard to find all maximal mutually commuting subsets of a set of columns $S$. We made this problem fixed-parameter-tractable by restricting to mutually commuting subsets of size at most $k$ with $k$ some fixed parameter. Here we take a different approach based on the idea that while it is hard to find all mutually commuting subsets of $S$ it is, given a column $v_i \in S$, tractable to find some maximal mutually commuting subset that contains $v_i$. Constructing this subset reduces to finding, given a graph $G$ and some vertex $v$, some clique that contains $v$. Note that we do not get to choose which clique will be found, only that one will be found. This can be done in time $O(M)$ where $M$ is the number of nodes in the graph $G$. Hence we can find, given a column $v_i$, a single maximal mutually commuting subset of $S$ that contains $v_i$. We can use only this set to evaluate whether a given other column $v_i$ can be written as the commuting composition of elements of $A$. The two subroutine functions are given in \cref{alg:greedy commuting independenceset,alg:greedy commuting checkindependence} and their complexity is analyzed in \cref{thm:greedy subroutine complexity}.
\begin{algorithm}[H]
\caption{CheckIndependence [greedy commuting]}
\label{alg:greedy commuting checkindependence}
\begin{algorithmic}[1]

\Require{A column $v_i$, a set of columns $S$, the column index $i$ of the column $v_i$ }

\Ensure{Boolean $a$}\\

\State // We will consistently write columns of the flow matrix $F$ as $v_i$ where $i$ indicates the column index of $v_i$ in $F$.\\

\State \textbf{Set} $a$ to FALSE

\For{columns $w_j \in S$}

	\State Construct maximal mutually commuting subset $C$ of $S$ containing $w_j$

	\For{$t \in [0:\mathrm{length}(p(v_i,i))-1]$}

		\For{all columns $c_l \in C $}

			\If{$p(c_l,l)_t =_w p(v_i,i)_t$ }

				\If{$t=\mathrm{length}(p(v_i,i))-1$}\\

					\State // Note that all the elements of $C$ commute. This means that if $p(v_i,i) =_w p(c_j,j)$ for a \State // particular $c_j$ it must also hold that $p(v_i,i) =_w p(c_k,k)$ for all $c_k \in C$. \\

					\State \textbf{Set} $a$ to TRUE

				\EndIf
			
			\Else

				\State Go to next column $w_j \in S$ at line 2

			\EndIf

		\EndFor

	\EndFor

\EndFor
\State\Return{$a$}

\end{algorithmic}	
\end{algorithm}

\begin{algorithm}[H]
\caption{DependenceSet [greedy commuting]}
\label{alg:greedy commuting independenceset}
\begin{algorithmic}[1]

\Require{A column $v_i$, a set of columns $S$, the column index $i$ of the column $v_i$ }

\Ensure{Set of columns $A_i$ that $v_i$ depends on}

\State \\ We will consistently write columns of the flow matrix $F$ as $v_i$ where $i$ indicates the column index of $v_i$ in $F$.

\For{columns $w_j \in S$}

	\State Construct a maximal mutually commuting subset $C$ of $S$ containing the column $w_j$

	\For{$t \in [0:\mathrm{length}(p(v_i,i))-1]$}

		\For{all columns $c_l \in C $}

			\If{$p(c_l,l)_t =_w p(v_i,i)_t$ }\\

				\State // Note that all the elements of $C$ commute. This means that if $p(v_i,i) =_w p(c_j,j)$ for a particular \State // $c_j$ it  must also hold that $p(v_i,i) =_w p(c_k,k)$ for all $c_k \in C$. \\

				\State Add $c_l$ to  $A_i$

				\If{$t=\mathrm{length}(p(v_i,i))-1$}

					\State\Return{$A_i$}

				\EndIf

			\Else

				\State \textbf{Set} $A_i$ to empty set

				\State Go to next column $w_j \in S$ at line 1

			\EndIf
		
		\EndFor

	\EndFor

\EndFor
\State\Return{$A_i$}

\end{algorithmic}
\end{algorithm}

\begin{theorem}\label{thm:greedy subroutine complexity}
The subroutines $\mathbf{CheckIndependence(~)}\mathrm{~[greedy]}$ and $\mathbf{IndependenceSet(~)}\mathrm{~[greedy]}$ have time complexity
\begin{align}
O\big(\mathbf{CheckIndependence(~)}\mathrm{~[greedy]}\big) &= O(N M^3)\\
O\big(\mathbf{IndependenceSet(~)}\mathrm{~[greedy]}\big) &= O(N M^3)
\end{align}
where $M = |S|$ the size of the independent set $S$ and $N$ is the length of the input column $v_i$ or equivalently the number of rows in the flow matrix $F$. Note that for our purposes we have $M\leq N$.
\end{theorem}
\begin{proof}\
We again find the time complexity of $\mathbf{CheckIndependence(~)}\text{~[greedy]}$ by a counting argument. The $\mathbf{For}$-loop on line 2 takes $O(M)$ time to iterate over. We argued above that the greedy construction on line 3 can be done in $O(M)$ time, the $\mathbf{For}$-loop on line 4 takes $O(N)$ time to iterate over, the $\mathbf{For}$-loop on line 5 takes $O(M)$ time to iterate over and the $\mathbf{If}$-clauses in the body can be evaluated in constant time. This means we get a total worst case time complexity of $O(N M^3)$. We can again make the same argument for $\mathbf{IndependenceSet(~)}\text{~[greedy]}$.
\end{proof}
\newpage
\newpage
\newpage
\newpage
\clearpage
\section{Surface code operation counts}\label{appsec:operation counts}
\begin{table}[h!]
\begin{tabular}{|c|c|c|c|c|c|c|c|c|c|c|c|c|c|c|c|c|c||c|}
\hline
Steps 				& 1 & 2 & 3& 4& 5 & 6 & 7& 8& 9 & 10 & 11& 12& 13 & 14 & 15& 16& 17 &$Z$-cycle Total  \\
\hline
\hline
$\sqrt{S}$ gate 	& \phantom{$0$} & \phantom{$0$} & \phantom{$0$}& \phantom{$0$}& $2$ & $2$ & \phantom{$0$} & \phantom{$0$} & \phantom{$0$} & $2$  & $2$ & \phantom{$0$} & \phantom{$0$} &\phantom{$0$} & \phantom{$0$}  & \phantom{$0$} & \phantom{$0$} & $8$ \\
\hline
$Z$ rotation		& \phantom{$0$} & \phantom{$0$} & \phantom{$0$}& \phantom{$0$}& $2$ & $2$ & \phantom{$0$} & \phantom{$0$} & \phantom{$0$} & $2$  & $1$ & \phantom{$0$} & \phantom{$0$} &\phantom{$0$} & \phantom{$0$}  & \phantom{$0$} & \phantom{$0$} & $7$  \\
\hline
Shuttling			& \phantom{$0$} & \phantom{$0$} & $1$& \phantom{$0$}& \phantom{$0$} & \phantom{$0$} & $1$ & $1$ & \phantom{$0$} & \phantom{$0$}  & \phantom{$0$} & $1$ & \phantom{$0$} &$2$ & $1$  & $2$ & $1$ & $10$  \\
\hline
Global rotation 	& \phantom{$0$} & $1$ & \phantom{$0$}& \phantom{$0$}& \phantom{$0$} & \phantom{$0$} & \phantom{$0$} & \phantom{$0$} & \phantom{$0$} & \phantom{$0$}  & \phantom{$0$} & \phantom{$0$} & $1$ &\phantom{$0$} & \phantom{$0$}  & \phantom{$0$} & \phantom{$0$} & $2$  \\
\hline
Measurement 		& \phantom{$0$} & \phantom{$0$} & \phantom{$0$}& \phantom{$0$}& \phantom{$0$} & \phantom{$0$} & \phantom{$0$} & \phantom{$0$} & \phantom{$0$} & \phantom{$0$}  & \phantom{$0$} & \phantom{$0$} & \phantom{$0$} &\phantom{$0$} & \phantom{$0$}  & $1$ & \phantom{$0$} & $1$  \\
\hline
\end{tabular}
\caption{Gate count per gate type and per step for the $Z$-ancilla during the $Z$-cycle of the surface code cycle described in \cref{sec:surface code implementation} and \cref{subsec:practical implementation}. Specifically the $Z$-ancilla is taken to be qubit $A$ in \cref{fig:unit cells} (right). Table cells that are left empty signify zero entries.
}
$\vspace{2mm}$\\
\begin{tabular}{|c|c|c|c|c|c|c|c|c|c|c|c|c|c|c|c|c|c||c|}
\hline
Steps 				& 1 & 2 & 3& 4& 5 & 6 & 7& 8& 9 & 10 & 11& 12& 13 & 14 & 15& 16&17&$X$-cycle Total  \\
\hline
\hline
$\sqrt{S}$ gate 	& \phantom{$0$} & \phantom{$0$} & \phantom{$0$}& \phantom{$0$}& \phantom{$0$} & \phantom{$0$} & \phantom{$0$} & \phantom{$0$} & \phantom{$0$} & \phantom{$0$}  & \phantom{$0$} & \phantom{$0$} & \phantom{$0$} &\phantom{$0$} & \phantom{$0$}  & \phantom{$0$} & \phantom{$0$} & $0$ \\
\hline
$Z$ rotation		& \phantom{$0$} & \phantom{$0$} & \phantom{$0$}& \phantom{$0$}& \phantom{$0$} & \phantom{$0$} & \phantom{$0$} & \phantom{$0$} & \phantom{$0$} & \phantom{$0$}  & \phantom{$0$} & \phantom{$0$} & \phantom{$0$} &\phantom{$0$} & \phantom{$0$}  & \phantom{$0$} & \phantom{$0$} & $0$  \\
\hline
Shuttling			& \phantom{$0$} & \phantom{$0$} & \phantom{$0$}& \phantom{$0$}& \phantom{$0$} & \phantom{$0$} & \phantom{$0$} & \phantom{$0$} & \phantom{$0$} & \phantom{$0$}  & \phantom{$0$} & \phantom{$0$} & \phantom{$0$} &$2$ & $1$  & \phantom{$0$} & $1$ & $4$ \\
\hline
Global rotation 	& \phantom{$0$} & $1$ & \phantom{$0$}& \phantom{$0$}& \phantom{$0$} & \phantom{$0$} & \phantom{$0$} & \phantom{$0$} & \phantom{$0$} & \phantom{$0$}  & \phantom{$0$} & \phantom{$0$} & $1$ &\phantom{$0$} & $1$  & \phantom{$0$} & \phantom{$0$} & $3$  \\
\hline
Measurement 		& \phantom{$0$} & \phantom{$0$} & \phantom{$0$}& \phantom{$0$}& \phantom{$0$} & \phantom{$0$} & \phantom{$0$} & \phantom{$0$} & \phantom{$0$} & \phantom{$0$}  & \phantom{$0$} & \phantom{$0$} & \phantom{$0$} &\phantom{$0$} & \phantom{$0$}  & $1$ & \phantom{$0$} & $1$  \\
\hline
\end{tabular}
\caption{Gate count per gate type and per step for the $Z$-ancilla during the $X$-cycle of the surface code cycle described in \cref{sec:surface code implementation} and \cref{subsec:practical implementation}. Specifically the $Z$-ancilla is taken to be qubit $A$ in \cref{fig:unit cells} (right). Table cells that are left empty signify zero entries. }
$\vspace{2mm}$\\
\begin{tabular}{|c|c|c|c|c|c|c|c|c|c|c|c|c|c|c|c|c|c||c|}
\hline
Steps				& 1 & 2 & 3& 4& 5 & 6 & 7& 8& 9 & 10 & 11& 12& 13 & 14 & 15& 16&17&$Z$-cycle Total  \\
\hline
\hline
$\sqrt{S}$ gate 	& \phantom{$0$} & \phantom{$0$} & \phantom{$0$}& \phantom{$0$}& $2$ & \phantom{$0$} & \phantom{$0$} & \phantom{$0$} & \phantom{$0$} & \phantom{$0$}  & $2$ & \phantom{$0$} & \phantom{$0$} &\phantom{$0$} & \phantom{$0$}    & \phantom{$0$} & \phantom{$0$} & $4$ \\
\hline
$Z$ rotation		& \phantom{$0$} & \phantom{$0$} & \phantom{$0$}& \phantom{$0$}& \phantom{$0$} & \phantom{$0$} & \phantom{$0$} & \phantom{$0$} & \phantom{$0$} & \phantom{$0$}  & \phantom{$0$} & \phantom{$0$} & \phantom{$0$} &\phantom{$0$} & \phantom{$0$}    & \phantom{$0$} & \phantom{$0$} & $0$ \\
\hline
Shuttling			& \phantom{$0$} & \phantom{$0$} & $1$& $1$& \phantom{$0$} & \phantom{$0$} & \phantom{$0$} & \phantom{$0$} & \phantom{$0$} & \phantom{$0$}  & \phantom{$0$} & \phantom{$0$} & \phantom{$0$} &\phantom{$0$} & \phantom{$0$}  & \phantom{$0$} & \phantom{$0$} & $2$  \\
\hline
Global rotation 	& \phantom{$0$} & $1$ & \phantom{$0$}& \phantom{$0$}& \phantom{$0$} & \phantom{$0$} & \phantom{$0$} & \phantom{$0$} & \phantom{$0$} & \phantom{$0$}  & \phantom{$0$} & \phantom{$0$} & $1$ &\phantom{$0$} & \phantom{$0$}    & \phantom{$0$} & \phantom{$0$} & $2$  \\
\hline
Measurement 		& \phantom{$0$} & \phantom{$0$} & \phantom{$0$}& \phantom{$0$}& \phantom{$0$} & \phantom{$0$} & \phantom{$0$} & \phantom{$0$} & \phantom{$0$} & \phantom{$0$}  & \phantom{$0$} & \phantom{$0$} & \phantom{$0$} &\phantom{$0$} & \phantom{$0$}    & \phantom{$0$} & \phantom{$0$} & $0$ \\
\hline
\end{tabular}
\caption{Gate count per gate type and per step for a data qubit during the $Z$-cycle of the surface code cycle described in \cref{sec:surface code implementation} and \cref{subsec:practical implementation}. Specifically the data qubit is taken to be qubit $1$ in \cref{fig:unit cells} (right) but other data qubits will have the same gate count up to a possible reordering of steps. Table cells that are left empty signify zero entries.}
$\vspace{2mm}$\\
\begin{tabular}{|c|c|c|c|c|c|c|c|c|c|c|c|c|c|c|c|c|c||c|}
\hline
Steps				& 1 & 2 & 3& 4& 5 & 6 & 7& 8& 9 & 10 & 11& 12& 13 & 14 & 15& 16&17&$X$-cycle Total  \\
\hline
\hline
$\sqrt{S}$ gate 	& \phantom{$0$} & \phantom{$0$} & \phantom{$0$}& \phantom{$0$}& \phantom{$0$} & $2$ & \phantom{$0$} & \phantom{$0$} & \phantom{$0$} & \phantom{$0$}  & \phantom{$0$} & $2$ & \phantom{$0$} &\phantom{$0$} & \phantom{$0$}    & \phantom{$0$} & \phantom{$0$} & $4$ \\
\hline
$Z$ rotation		& \phantom{$0$} & \phantom{$0$} & \phantom{$0$}& \phantom{$0$}& \phantom{$0$} & \phantom{$0$} & \phantom{$0$} & \phantom{$0$} & \phantom{$0$} & \phantom{$0$}  & \phantom{$0$} & \phantom{$0$} & \phantom{$0$} &\phantom{$0$} & \phantom{$0$}    & \phantom{$0$} & \phantom{$0$} & $0$ \\
\hline
Shuttling			& \phantom{$0$} & \phantom{$0$} & \phantom{$0$}& \phantom{$0$}& \phantom{$0$} & \phantom{$0$} & \phantom{$0$} & $1$& $1$ & \phantom{$0$}  & \phantom{$0$} & \phantom{$0$} & \phantom{$0$} &\phantom{$0$} & $1^*$  & \phantom{$0$} & $1^*$ & $4$  \\
\hline
Global rotation 	& \phantom{$0$} & $1$ & \phantom{$0$}& \phantom{$0$}& \phantom{$0$} & \phantom{$0$} & \phantom{$0$} & \phantom{$0$} & \phantom{$0$} & \phantom{$0$}  & \phantom{$0$} & \phantom{$0$} & $1$ &\phantom{$0$} & \phantom{$0$}    & \phantom{$0$} & \phantom{$0$} & $2$  \\
\hline
Measurement 		& \phantom{$0$} & \phantom{$0$} & \phantom{$0$}& \phantom{$0$}& \phantom{$0$} & \phantom{$0$} & \phantom{$0$} & \phantom{$0$} & \phantom{$0$} & \phantom{$0$}  & \phantom{$0$} & \phantom{$0$} & \phantom{$0$} &\phantom{$0$} & \phantom{$0$}    & \phantom{$0$} & \phantom{$0$} & $0$ \\
\hline
\end{tabular}
\caption{Gate count per gate type and per step for a data qubit during the $X$-cycle of the surface code cycle described in \cref{sec:surface code implementation} and \cref{subsec:practical implementation}. Specifically the data qubit is taken to be qubit $1$ in \cref{fig:unit cells} (right) but other data qubits will have the same gate count up to a possible reordering of steps. Table cells that are left empty signify zero entries.\\
$^*$: Only half of the data qubits move during this step. In the total gate count this gate is counted towards all data qubits. }
\end{table}

\end{document}